\documentclass[a4paper,11pt]{amsart}
\usepackage{amsthm,amsmath}
\usepackage{graphicx,color}%caption,subcaption

\usepackage[normalem]{ulem}

\newtheorem{claim}{Claim}[section]
\newtheorem{theorem}[claim]{Theorem}
\newtheorem{proposition}[claim]{Proposition}
\newtheorem{lemma}[claim]{Lemma}
\newtheorem{remark}[claim]{Remark}
\newtheorem{example}[claim]{Example}
\newtheorem{definition}[claim]{Definition}
\newtheorem{corollary}[claim]{Corollary}

\newcommand{\ds}{\displaystyle}
\newcommand{\eqskip}{ \vspace*{2mm}\\ }

\newcommand{\soutb}{\bgroup\markoverwith{\textcolor{blue}{\rule[.5ex]{2pt}{1pt}}}\ULon}
\newcommand{\soutr}{\bgroup\markoverwith{\textcolor{red}{\rule[.5ex]{2pt}{1pt}}}\ULon}
\renewcommand{\Re}{\text{\rm Re}\,}
\renewcommand{\Im}{\text{\rm Im}\,}
\DeclareMathOperator*{\esssup}{ess\,sup}
\DeclareMathOperator*{\essinf}{ess\,inf}

\begin{document}
\title[Damped wave equation on metric graphs]{Eigenvalue asymptotics for the damped wave equation on metric graphs}
\author[Correspond author]{Pedro Freitas}
\address{Department of Mathematics, Faculty of Human Kinetics, Universidade de Lisboa
\and
Group of Mathematical Physics, Faculty of Sciences, Universidade de Lisboa, Campo Grande, Edif\'{\i}cio C6
1749-016 Lisboa, Portugal }
\email{psfreitas@fc.ul.pt}
\author{Ji\v{r}\'{\i} Lipovsk\'{y}}
\address{Department of Physics, Faculty of Science, University of Hradec Kr\'{a}lov\'{e},
Rokitansk\'eho 62, 500\,03 Hradec Kr\'{a}lov\'{e}, Czechia}
\email{jiri.lipovsky@uhk.cz}

\thanks{The second author was partially supported by project Development activities postdocs at the University of the
Hradec Kr\'alov\'e CZ.1.07/2.3.00/30.0015
and project 15-14180Y ``Spectral and resonance properties of quantum models'' of the Czech Science Foundation.
Both authors were partially supported by FCT (Portugal) through project PTDC/\ MAT-CAL/\ 4334/2014.}

\date{\today}

\begin{abstract} We consider the linear damped wave equation on finite metric graphs and analyse its spectral
properties with an emphasis on the asymptotic behaviour of eigenvalues. In the case of equilateral graphs
and standard coupling conditions we show that there is only a finite number of high-frequency abscissas, whose
location is solely determined by the averages of the damping terms on each edge.
We further describe some of the possible behaviour when the edge lengths are no longer necessarily equal but remain
commensurate. 
\end{abstract}


\maketitle

\section{Introduction}

The simplest model for an inhomogeneous damped vibrating string is given by the equation
\begin{equation}\label{waveeq1d}
  \partial_{tt}u(t,x) + 2 a(x) \partial_{t}u(t,x) = \partial_{xx}u(t,x) + b(x)u(t,x),
\end{equation}
complemented with initial and boundary conditions at the end points. Although quite simple, the explicit
dependence of the damping on the space variable makes the problem genuinely non-selfadjoint and is responsible for several
interesting properties which have received much attention in the literature within the last twenty years -- see~\cite{BoF,GH}
and the references therein. One particular aspect which is of interest is the characterisation of the spectrum and, within that scope,
the behaviour of the high frequencies. In the case of equation~\eqref{waveeq1d} with vanishing $b$, it was first suggested by
the formal calculations in~\cite{cfns} that the real part of these high frequencies would cluster around minus the average
of the damping term as the corresponding imaginary parts converge to infinity. This was proven rigorously in~\cite{CZ} where
the first two terms in the asymptotic expansion of the eigenvalues were computed, while the complete asymptotic expansion was obtained
in~\cite{BoF} for the first time, including the explicit determination of its first four terms. Other models for the linear
damped wave equation include, for instance, imposing the damping on the boundary~\cite{CZ2,AMN}. Our approach may, in principle,
also be extended to such problems.

In the case of more complex structures involving several segments which are joined at the endpoints, it makes sense to model
each component by an equation of the form~\eqref{waveeq1d} with corresponding potentials and damping functions and have either
a coupling at the common vertices or some boundary condition such as Dirichlet at the isolated endpoints. Indeed, there is a vast
literature on the topic of waves on networks of strings as may be seen in the review~\cite{Z}. However, most of this work revolves
around the observability and controllability of such problems and, in particular, this means that in most cases the problem may still
be reduced to the study of a self-adjoint (vector) operator. One exception to this is given by~\cite{AMN2}, where the wave equation
with indefinite sign damping is considered on a star graph, generalising the results in~\cite{fr96}.

A different starting point which is also connected to the undamped case comes from the quantum graph literature where the resulting
problem, although including the explicit dependence on the spatial variable via a potential playing a similar role to $b$ in
equation~\eqref{waveeq1d} above, is still reducible to a self--adjoint operator for a wide variety of coupling conditions
(see \emph{e.g.} \cite{AMR,Ku1, KS, RS}). Note that, due to the fact that the damping will not, in general, be continuous
across vertices, the asymptotic behaviour of the eigenfunctions is not necessarily the same as that of the undamped
problem -- see Remark~\ref{remnotcont}.

The main object of study in the present paper is the wave equation on graphs with potential and damping functions depending
explicitly on the space variable and with coupling conditions at the vertices. We are interested in understanding the
asymptotic behaviour of the associated spectrum and, in particular, the counterpart for the case of graphs of the result
on the spectral abscissa mentioned above for one segment. This is related to the rate of decay of high-frequency modes on such
structures and is therefore an important issue in applications. We remark that, as already noted in~\cite{Z}, this behaviour
depends in a subtle way on topological and number theoretic properties of the network, without it being possible to understand
the global dynamics by simply studying each component in isolation from the rest. 

Our main contribution here is to describe the possible asymptotic distributions of the high frequencies and show that, under certain
conditions on the edge lengths and the coupling conditions, there exists at most a finite number of values around which the real parts of the
eigenvalues may cluster, as the imaginary parts grow to infinity -- see Theorems~\ref{thm-onlyaverage} and~\ref{alla}, the results in
Section~\ref{sec-number} (in particular, Theorem~\ref{thm-main}) and the examples in the last section. In particular, we show
that under such conditions these {\it high-frequency
spectral abscissas} are the same as those obtained for the same graph where the damping terms are replaced by their averages on each edge
and with zero potentials. Furthermore, they may be computed by solving a polynomial equation (see Remark~\ref{rem-polyeq}
and the examples in Section~\ref{sec-examples}).

While in the case of $N$ equilateral edges we show that the number of such abscissas is at most $2N$ and that, if the graph is bipartite
(does not have cycles of odd length), this number cannot be larger than $N$, we also show that these results are sharp in the sense that for
certain classes of graphs there exist damping terms yielding these maximal number of abscissas. If we drop the assumption that all edges
have equal lengths and assume instead that they are commensurate, we give examples which indicate that the number of these abscissas may now
be as large as desired.

The damped wave equation on metric graphs may thus be viewed as a bridge in complexity between the scalar damped wave equation in one and two
spatial dimensions. While in the former case there is only one high-frequency spectral abscissa corresponding to the averaging of the damping
on a segment, in the two-dimensional case there are several different behaviours for the high frequencies corresponding to the different
trajectories on the domain~\cite{BLR,AL,Sj}. On a metric graph, these trajectories are restricted to the edges, with the damping being averaged
along each edge and these averages then being combined via the topology of the graph. While the structure of the spectrum for equilateral
graphs is still fairly simple as described above, we see that this is no longer the case in general. 

The paper is structured as follows. In the next section we describe the model and the corresponding notation. The third section presents some
basic asymptotic properties of the secular equation and vertex scattering matrix for eigenvalues with high imaginary part, essentially
following the approach in~\cite{BoF} for a single interval. The number and location of high-frequency abscissas are studied in Sections~\ref{sec-loc}
and~\ref{sec-number}. It turns out that to obtain the results in this last section it is more convenient to use the method of pseudo-orbit
expansions to handle the secular equation and we describe this approach in Section~\ref{sec-orbit}.
In Section~\ref{sec-examples} several examples are provided illustrating the previous results. 

\section{Description of the model}\label{sec-description}
We consider a compact metric graph $\Gamma$ with $N<\infty$ finite edges $\{e_j\}_{j=1}^N$ of lengths $\{l_j\}_{j=1}^N$. On each edge $e_{j}$
a linear damped wave equation of the form
%%-----
\begin{equation}
  \partial_{tt} w_j(t,x) + 2 a_j(x) \partial_{t} w_j(t,x) = \partial_{xx} w_j(t,x) + b_j(x)w_j(t,x) \label{eq-wetime}
\end{equation}
%%-----
is considered, with both the damping functions $a_j(x)$ and the potentials $b_j(x)$ real and bounded. The above set of partial differential equations
can be rewritten in a more elegant way
%%-----
$$
  \partial_{t} \begin{pmatrix}\vec{w_0}(t,x)\\ \vec{w_1}(t,x) \end{pmatrix} = \begin{pmatrix}0 & I\\ I \frac{\mathrm{d}^2}{\mathrm{d}x^2}+ B & -2 A
\end{pmatrix} \begin{pmatrix}\vec{w_0}(t,x)\\ \vec{w_1}(t,x)\end{pmatrix}\,,
$$
%%-----
where $I$ is the $N\times N$ identity matrix, $A$ and $B$ are diagonal $N\times N$ matrices with entries $a_j$ and $b_j$, respectively, $\vec{w_0}$ is
the vector of functions $w_j$ and $\vec{w_1}$ its time derivative.

The Ansatz $w_j(t,x) = \mathrm{e}^{\lambda t} u_j(x)$ allows us to rewrite equation~(\ref{eq-wetime}) as
%%-----
\begin{eqnarray}
  \partial_{xx} u_j(x) -(\lambda^2+2\lambda a_j(x)-b_j(x))u_j(x) = 0 \label{eq-ldw1}
\end{eqnarray}
%%-----
and therefore reformulate the time dependent problem as a spectral problem. More precisely, in this way we are led to the operator
\[
H = \begin{pmatrix}0 & I\\ I
\frac{\mathrm{d}^2}{\mathrm{d}x^2}+ B & -2 A \end{pmatrix},
\]
with domain consisting of functions $(\psi_1(x),\psi_2(x))^\mathrm{T}$ with components
of both $\psi_1$ and $\psi_2$ in  $W^{2,2}(e_j)$ for the corresponding edge and satisfying suitable coupling conditions at the vertices.

Now we clarify the term ``suitable coupling conditions''; we construct all operators $H$ satisfying the above properties, for which $i H_{A = 0}$ is
self-adjoint. Although we shall not impose any restrictions on the
sign of the potentials in the paper, to motivate the coupling conditions, we start in this
section with all potentials $b_j(x)$ being non-positive.
We define $\psi = (\psi_1,\psi_2)^{\mathrm{T}}$, $\phi = (\phi_1,\phi_2)^{\mathrm{T}}$, denote by $(\cdot,\cdot)$ the $L^2(\Gamma)$
scalar product linear in the second argument and antilinear in the first one and by $((\cdot,\cdot))$ the $L^2(\Gamma)\oplus L^2(\Gamma)$ scalar
product. For the condition for selfadjointness of $iH$ we obtain 
%%-----
$$
  0 = ((\psi, i H_{A = 0} \phi)) - ((i H_{A = 0} \psi,\phi)) = i \left[((\psi, H_{A = 0} \phi))+ ((H_{A = 0} \psi, \phi)) \right]
$$
%%-----
and hence
%%-----
$$
  0 = (\psi_1,\phi_2)+ (\psi_2,\phi_1'') + (\psi_2,B\phi_1)+ (\psi_2,\phi_1)+ (\psi_1'',\phi_2)+ (B \psi_1,\phi_2)\,.
$$
%%-----
With use of $\phi_2 = \lambda \phi_1$ with $\lambda$ purely imaginary and the fact that $B$ is real and diagonal one rewrites the above equation to
%%-----
$$
  0 = (\psi_1,\phi_1)- (\psi_1,\phi_1'') - (\psi_1,B\phi_1)- (\psi_1,\phi_1)+ (\psi_1'',\phi_1)+ (\psi_1,B\phi_1)\,,
$$
%%-----
now integrating by parts we obtain 
%%-----
\begin{equation}
  \sum_e \left[\bar{\psi_1}'\phi_1-\bar{\psi_1}\phi_1' \right]_0^{l_e} = 0\,. \label{eq-coupling2}
\end{equation}
%%-----
Denoting the vectors of functional values and the outgoing derivatives at the vertices by $\Psi$ and $\Psi'$, respectively, and choosing $\Psi =
\Phi$,
the above condition implies $\Psi'^*\Psi - \Psi^*\Psi' = 0$ for all $\Psi$, where $^*$ stands for hermitian conjugation. It is straightforward
to check that this is equivalent to $\|\Psi+ i \Psi'\| = \|\Psi- i \Psi'\|$, leading to the general coupling condition
%%-----
\begin{equation}
  (U-I) \Psi + i (U+I)\Psi' = 0\,,\label{eq-coupling}
\end{equation}
%%-----
for any $N\times N$ unitary matrix $U$ and where $I$ is the $N\times N$ identity matrix. Specific examples of unitary matrices $U$ are given
below and in Section~\ref{sec-examples}. Let us note that the more general choice $\Psi\pm
i l \Psi'$ instead of $\Psi\pm i \Psi'$ with any real $l$ is possible; however, this leads to the condition $(U_l-I) \Psi + il (U_l+I)\Psi' = 0$ and
the simple transformation between $U_l$ and $U$ 
$$
  U_l = [(U+I)-l(U-I)]^{-1}[(U+I)+l(U-I)]
$$
shows that this equation is, in fact, equivalent to~(\ref{eq-coupling}).

The reverse implication that any pair of vectors $(\Phi^\mathrm{T},{\Phi'}^\mathrm{T})^\mathrm{T}$ and $(\Psi^\mathrm{T},{\Psi'}^\mathrm{T})^\mathrm{T}$
satisfying the coupling condition~(\ref{eq-coupling}) (if their components are substituted for $\Psi$ and $\Psi'$), 
also satisfies~(\ref{eq-coupling2}), follows from properties of unitary matrices. For more details on the above construction we refer e.g. to 
\cite{FT} or Proposition~17.1.2 in the book \cite{BEH}.

In agreement with the notation for quantum graphs we shall denote the coupling for which functional values are equal at each vertex and the sum of
outward derivatives vanishes by \emph{standard coupling} (the terms Kirchhoff's, Neumann or free coupling may also be found in the literature). The
corresponding unitary coupling matrix is $U  = 2/d\, J - I$, with $d$ being the degree of the given vertex, $I$ the $d\times d$ identity matrix and
$J$ the $d\times d$ matrix with all entries equal to $1$. For instance, the forms of the coupling matrices for standard coupling are for
$d=2$ and $d=3$ the following
$$
  U_{d=2} = \begin{pmatrix}0& 1\\1&0\end{pmatrix}\,,\quad U_{d=3} = \frac{1}{3}\begin{pmatrix}-1& 2 & 2\\ 2 & -1 & 2\\ 2 & 2 & -1\end{pmatrix}\,.
$$

Moreover, we define coupling conditions for the boundary vertices. The \emph{Dirichlet} coupling denotes the situation when the functional value 
vanishes at the vertex, while in the case of \emph{Neumann} coupling the derivative at the vertex is zero. The general condition is called \emph{Robin} 
and it is described by equation~(\ref{eq-coupling}) with $U$ being a complex number of modulus one. This condition also includes the Dirichlet
($U = -1$) and Neumann ($U = 1$) conditions.

\section{The secular equation and asymptotic properties of eigenvalues and eigenfunctions}\label{sec-seceq}
In this section we establish the asymptotic properties of the fundamental system of solutions
and eigenvalues of $H$. Let us first consider a graph with all edges of
length one. 

\begin{theorem}\label{thm-justif}
Consider an equilateral graph with $N$ edges of unit length with the coupling between vertices given by a matrix $U$ as above.
Denote the damping and
potential on each edge by $a_{j}$ and $b_{j}$, respectively, and assume 
$a_j\in \mathcal{C}^{1}([0,1])$ and $b_j \in \mathcal{C}^0([0,1])$. Then there exist a positive real number $K_0$ such that for  
$K> K_0$ if $\lambda = r + i K$ is an eigenvalue, then $\lambda + 2\pi i + \mathcal{O}(1/K) $ is also an eigenvalue. Similarly, 
if $\lambda = r - i K$ is an eigenvalue, then $\lambda - 2\pi i + \mathcal{O}(1/K) $ is also an eigenvalue. This means that there exist 
sequences of eigenvalues $\lambda_{sn}$ satisfying
\begin{equation}
  \lambda_{sn} = 2\pi i n + c_0^{(s)}+ \mathcal{O}\left(\frac{1}{n}\right)\label{eq-asymptot}
\end{equation}
as $n$ goes to infinity, where the complex constants $c_0^{(s)}$ are, in general, different for each sequence $s$.
\end{theorem} 

The proof of this theorem relies partially on the results in~\cite{BoF} for a single segment. Of particular interest 
to us here is the following result giving a fundamental system of solutions on the interval $[0,1]$ which we reproduce 
here for completeness.

%%-----
\begin{lemma}[\cite{BoF}] \label{lem-denis}
Let $a \in \mathcal{C}^{m+1}[0,1]$ and $b \in \mathcal{C}^{m}[0,1]$. Then there exist two linearly independent solutions $u_\pm(x,\lambda)$ of
equation (\ref{eq-ldw1}) satisfying the initial condition $u_\pm(0,\lambda) = 1$ having the asymptotics
%%-----
\begin{eqnarray}
  u_{\pm}(x,\lambda) = \mathrm{e}^{\pm \lambda x \pm \int_{0}^x \phi_{\pm}(\xi,\lambda)\,\mathrm{d}\xi} \label{eq-asymp}
\end{eqnarray}
%%-----
in the $C^2[0,1]$ norm as $\mathrm{Im\,}\lambda\to \infty$ with
%%-----
\begin{eqnarray}
  \phi_{\pm}(x,\lambda) = \sum_{i = 0}^m \frac{\phi^{\pm}_{i}(x)}{\lambda^i}+ \mathcal{O}(\lambda^{-m-1}) \label{eq-phi}
\end{eqnarray}
%%-----
and
%%-----
\begin{eqnarray*}
  \phi_{0}^{(\pm)}(x) = a(x)\,,\quad \phi_{1}^{(\pm)}(x) = -\frac{1}{2} (\pm a'(x) +a^2(x) +b(x))\,,\\
  \phi_{i}^{(\pm)}(x) = -\frac{1}{2} \left( \pm {\phi'}^{(\pm)}_{i-1} +  \sum_{s=0}^{i-1} {\phi}_{s}^{(\pm)} \phi_{i-s-1}^{(\pm)}  \right) \,.
\end{eqnarray*}
%%-----
\end{lemma}
%%-----

\begin{definition}
We denote by $\bar{a}_{j}$ the average of the damping function $a_{j}(x)$ on the $j$-th edge, that is
$$
  \bar{a}_j = \frac{1}{l_j}\int_0^{l_j} a_j(x)\,\mathrm{d}x\,. 
$$  
\end{definition}

\begin{proof}[Proof of Theorem~\ref{thm-justif}]
Now we find a secular equation for a graph with edges of length one with coupling given by~(\ref{eq-coupling}). We construct a
flower-like model which is similar to the case studied for quantum graphs \cite{Ku1, EL2, DEL}. The main idea, which we shall now describe, is very
simple. For a general graph with $N$ edges of unit length, one considers a one-vertex model with $N$ loops, also of unit length; the coupling matrix
$U$ is, in a suitably chosen basis, block diagonal with blocks corresponding to the vertex coupling matrices. The Hamiltonian on the original graph is
unitarily equivalent to the Hamiltonian on the flower-like graph and hence every such graph can be described by this model. On each of the edges of
the graph we take linear combinations of solutions of the form~(\ref{eq-asymp}) as described below; the indices $j$ denoting a particular edge
are added where needed.

From Lemma~\ref{lem-denis} we have that on each edge there exist two linearly independent solutions $u_{\pm}(x,\lambda)$.
The general solution on each edge is thus given by $u_j(x) = \alpha_j u_+(x,\lambda)+ \beta_j u_-(x,\lambda)$
while, using the fact that the entries of $\Psi$ are $u_j(0)$ and $u_j(1)$ and those
of $\Psi'$ are $u_j'(0)$ and $-u_j'(1)$, the coupling condition (\ref{eq-coupling}) then becomes
\begin{multline*}
  0 = (U-I)\begin{pmatrix}\alpha_1 + \beta_1\\ \alpha_1 \mathrm{e}^{\lambda+\bar\phi_{+}^{(1)}}+\beta_1 \mathrm{e}^{-\lambda-\bar\phi_{-}^{(1)}}\\ \alpha_2+\beta_2\\ \vdots\\  \alpha_N \mathrm{e}^{\lambda+\bar\phi_{+}^{(N)}}+\beta_N \mathrm{e}^{-\lambda-\bar\phi_{-}^{(N)}}\end{pmatrix} +
  \\ 
  +i (U+I) \begin{pmatrix}\alpha_1(\lambda+ \phi_{+}^{(1)}(0))-\beta_1(\lambda+\phi_{-}^{(1)}(0))\\ -\alpha_1(\lambda+\phi_{+}^{(1)}(1))\mathrm{e}^{\lambda+\bar\phi_{+}^{(1)}}+\beta_1(\lambda+\phi_{-}^{(1)}(1))\mathrm{e}^{-\lambda-\bar\phi_{-}^{(1)}}\\ \alpha_2(\lambda+ \phi_{+}^{(2)}(0))-\beta_2(\lambda+\phi_{-}^{(2)}(0))\\\vdots\\ -\alpha_N(\lambda+\phi_{+}^{(N)}(1))\mathrm{e}^{\lambda+\bar\phi_{+}^{(N)}}+\beta_N(\lambda+\phi_{-}^{(N)}(1))\mathrm{e}^{-\lambda-\bar\phi_{-}^{(N)}}\end{pmatrix}\,.
\end{multline*}
Here the superscript of $j$ in $\phi_{\pm}^{(j)}$ distinguishes the edge (i.e. it is not the index of the expansion in Lemma~\ref{lem-denis}),
while $\bar\phi_{\pm}^{(j)} =\int_{0}^1 \phi_{\pm}^{(j)}(\xi,\lambda)\,\mathrm{d}\xi$ denotes the average of $\phi_{\pm}$ on the $j$-th edge.

Therefore, the secular equation 
may be written as
\begin{equation}
  \mathrm{det\,}[(U-I)M_1(\lambda)+ i (U+I) M_2(\lambda)] = 0\,,\label{eq-seceq1}
\end{equation}
where
$M_1$ is a block-diagonalizable matrix consisting of blocks of the form
$$
\begin{pmatrix}
    1 & 1\\
    \mathrm{e}^{\lambda+\bar\phi_{+}^{(j)}} & \mathrm{e}^{-\lambda-\bar\phi_{-}^{(j)}}
\end{pmatrix}
$$
and similarly $M_2$ consists of blocks of the form
$$
\begin{pmatrix}
    \lambda+ \phi_{+}^{(j)}(0) & -(\lambda+\phi_{-}^{(j)}(0))\\
    - (\lambda+\phi_{+}^{(j)}(1))\mathrm{e}^{\lambda+\bar\phi_{+}^{(j)}} & (\lambda+\phi_{-}^{(j)}(1))\mathrm{e}^{-\lambda-\bar\phi_{-}^{(j)}} 
\end{pmatrix}
$$

Since we assume $a_j\in \mathcal{C}^1([0,1])$, by Lemma~\ref{lem-denis} $\phi_1^{(\pm)(j)}$ is uniformly bounded and we have $\bar\phi_{\pm}^{(j)} = \bar a_j + \mathcal{O}(1/\lambda)$. 
The form of the secular equation is thus
\begin{multline*}
  P_{1}\mathrm{e}^{\lambda+\bar a_1+\lambda+\bar a_2+\dots +\lambda+\bar a_N+\mathcal{O}(1/\lambda)} + 
  P_{2}\mathrm{e}^{-\lambda-\bar a_1+\lambda+\bar a_2+\dots +\lambda+\bar a_N+\mathcal{O}(1/\lambda)} + 
  \\
  P_{3}\mathrm{e}^{\lambda+\bar a_1-\lambda-\bar a_2+\dots +\lambda+\bar a_N+\mathcal{O}(1/\lambda)} + 
  P_{4}\mathrm{e}^{-\lambda-\bar a_1-\lambda-\bar a_2+\dots +\lambda+\bar a_N+\mathcal{O}(1/\lambda)} +
  \\
   \dots +
  P_{2^N}\mathrm{e}^{-\lambda-\bar a_1-\lambda-\bar a_2-\dots -\lambda-\bar a_N+\mathcal{O}(1/\lambda)} = 0\,,
\end{multline*}
where each $P_q$ is a polynomial in $\lambda$ of degree $2N$. Since in the equation we must account for all possible combinations of signs of the
terms $(\lambda+a_j)$ in the exponents, there are $2^N$ terms.

For nontrivial $A$ the operator $H$ is a bounded perturbation of the operator $H_{A = 0}$ for $A$ equal
to zero. Thus according to Theorem~3.17 in \cite{Ka}, $H$ has infinitely many eigenvalues with no finite accumulation point. 
Let us assume $\lambda = r + i K$ with big $K$. We have
$$
  \frac{1}{\lambda} = \frac{r - Ki}{r^2+ K^2} = \mathcal{O}\left(\frac{1}{K}\right)\,.
$$
The secular equation can be written using the leading term of asymptotics in $K$
\begin{multline}
  F(\lambda) = K^{2N} \left(c_1 \mathrm{e}^{\lambda+\bar a_1+\lambda+\bar a_2+\dots+ \lambda+\bar a_N}+ 
  c_2 \mathrm{e}^{-\lambda-\bar a_1+\lambda+\bar a_2+\dots+ \lambda+\bar a_N}+ \dots \right.
  \\
  \left.c_Q \mathrm{e}^{-\lambda-\bar a_1-\lambda-\bar a_2-\dots- \lambda-\bar a_N}\right) \left(1+\mathcal{O}\left(\frac{1}{K}\right)\right) = 0  \label{eq-leading_term}
\end{multline}

Let $F(\lambda)= K^{2N} F_1(\lambda)+ K^{2N-1} F_2(\lambda)$ and $F(\lambda_0) = 0$, $F_1(\lambda_0') = 0$. 
Since $a_j \in \mathcal{C}^1([0,1])$ we know that $\phi_1^{(\pm)(j)}(x)$, and hence also the second term of the secular equation are uniformly bounded,
say, $|F_2(\lambda)|<\tilde K$.
Hence $|F_1(\lambda_0)|<\tilde K/K$ and since $F_1(\lambda)$ is continuous, $\lambda_0'$ approaches $\lambda_0$ as $K\to \infty$.  
From (\ref{eq-leading_term}) we have that $F_1$ is equal to zero for $\lambda_0'+2\pi i$. Hence there exists such
$\lambda = \lambda_0 + 2\pi i + \mathcal{O}(1/\lambda_0)$ for which $F(\lambda)=0$ and there exists a sequence of
eigenvalues of the form~\eqref{eq-asymptot}. For $\lambda = r -i K$ the proof is similar.
\end{proof} 
Another approach to constructing the secular equation based on pseudo-orbit expansions will be given in Section~\ref{sec-orbit} below. 

The real parts of the coefficients $c_0^{(s)}$ are the subject of the following definition.

\begin{definition}{\rm
We say that $\omega_0$ is a \emph{high-frequency abscissa} of the operator $H$ if there exists a sequence of eigenvalues of
$H$, say $\{\lambda_{sn}\}_{n=1}^\infty$, such that
\[
 \lim_{n\to \infty}\mathrm{Im\,}\lambda_{sn} = \pm \infty \mbox{ and } \lim_{n\to \infty}\mathrm{Re\,}\lambda_{sn} = \omega_0.
\]
The number of distinct high-frequency abscissas of $H$ will be called its \emph{abscissa count} $\alpha_c$.
}\end{definition}

Now we present a theorem which generalises previous results for the localisation of the high-frequency abscissa on a single segment
(see~~\cite{CZ,cfns}, for instance) to the case of equilateral graphs. More precisely, it states that in order
to determine the localisation of high-frequency abscissas of such graphs, it is enough
to consider the average of the damping coefficients on each edge.
%%-----
\begin{theorem}\label{thm-onlyaverage}
Let $\Gamma$ be an equilateral graph with $N$ edges of lengths $l_j= l_0$, $j = 1, \dots, N$, with the coupling conditions
given by the coupling matrix $U$. Let the damping and potential functions $a_j(x)$ and $b_j(x)$ be bounded and continuous
on each edge. Let $\lambda_{sn}$ be
eigenvalues of the corresponding problem (\ref{eq-ldw1}) and $\mu_{sn}$ eigenvalues for $a_j$ replaced by its average on each edge and $b_j = 0$.
Then the constant terms $c_0^{(s)}$ in the asymptotic expansion (\ref{eq-asymptot}) for each of the sequences $\lambda_{sn}$ 
coincide with the corresponding constant terms in the asymptotic expansion of $\mu_{sn}$.
\end{theorem}

Before proving the above theorem we state a result relating the eigenvalues of an equilateral graph with
edges of unit length to those of an equilateral graph with scaled lengths of the edges.
%%----
\begin{lemma}\label{lem-l0}
Let $\lambda$ be an eigenvalue of an equilateral graph $\Gamma$ with edges of unit length, damping coefficients $a_j(x)$,
potentials $b_j(x)$, $x\in (0,1)$ and the coupling given by $U$. Then $\lambda /l_0$ is eigenvalue of the same graph with damping coefficients 
$\tilde a_j(y) = a_j(y/l_0)/l_0$, potentials $\tilde b_j(y) = b_j(y/l_0)/l_0^2$, $y \in (0,l_0)$,
all edges of length $l_0$ and coupling given by
\[ 
 U_{l_0} = [(l_0 -1) U + (l_0+1)I]^{-1}[(l_0+1)U + (l_0-1)I].
\]
\end{lemma}
\begin{proof}
The wavefunction component on $j$-th edge of the first graph satisfies
%%----
$$
  \partial_{xx} u_j(x,\lambda) - (\lambda^2 + 2 \lambda a_j(x) - b_j(x)) u_j (x,\lambda) = 0\,,\quad x \in (0,1)\,.
$$
%%----                                                                 
Substituting $x = y/l_0$, rewriting the damped wave equation for $v_j(y) = u_j(y/l_0)$ and using 
$\partial_{xx} u_j(x,\lambda) = l_0^2 \partial_{yy} u_j\left(y/l_0,\lambda\right)$ we obtain
%%----
$$
  \partial_{yy} u_j(y/l_0,\lambda) - \frac{1}{l_0^2}(\lambda^2 + 2 \lambda a_j(y/l_0) - b_j(y/l_0)) u_j (y/l_0,\lambda) = 0\,,\quad y \in (0,l_0)\,
$$
%%----
from which one finds the eigenvalues, damping and potential. Since $v_j(l_0) = u_j(1)$ and $v_j'(l_0) = u_j'(1)/l_0$ one finds 
%%----
$$
 U_{l_0} - I =  C (U-I)\,,\quad U_{l_0} + I = l_0 C (U+I)
$$
%%----
with $C$ being regular square matrix and hence one finds corresponding coupling matrix $U_{l_0}$.
\end{proof}
%%----

\begin{remark}{\rm
Note that Theorem~\ref{thm-onlyaverage} and Lemma~\ref{lem-l0} can be used also for graphs with commensurate lengths. We may introduce
vertices of degree $2$
at the distance $l_0$ and obtain an equilateral graph. The averaging is done only on the sub-edges.
}\end{remark}

\begin{proof}[Proof of Theorem~\ref{thm-onlyaverage}]
Let us first consider the case where all lengths are equal to~$1$.
We shall now separate the spectrum of the unitary matrix $U$ into $-1, +1$ and the remaining non-real eigenvalues on the unit 
circle. Denoting by $n_{-}$ and $n_{+}$ the number of $-1$'s and $+1$'s in the spectrum of $U$, respectively, we may write
\begin{equation}
U =
V^{-1}\begin{pmatrix}-I_{n_{-}}&0 & 0\\ 0 &
I_{n_{+}}&0 \\ 0& 0 & D\end{pmatrix}V \label{eq-u}
\end{equation}
with $V$ being a unitary matrix and $I_{k}$ being the $k\times k$ identity matrix and with the diagonal matrix 
$D = \mathrm{diag\,}(\mathrm{e}^{i\varphi_1},\dots , \mathrm{e}^{i\varphi_{2N-n_{-}-n_{+}}})$ containing the eigenvalues of $U$ other than $\pm 1$.
The leading term of the high-frequency asymptotics of the secular equation, obtained by substituting (\ref{eq-asymptot}) to (\ref{eq-seceq1}), is given by 
%%-----
$$
  (2 \pi i n)^{2N-n_{-}}\,\mathrm{det\,}\left[\begin{pmatrix}-2 I & 0 & 0\\ 0 & 0 & 0\\ 0 & 0 & 0\end{pmatrix} M_3 + i \begin{pmatrix} 0 & 0 & 0\\ 0
& 2 I & 0 \\ 0 & 0 & D + I \end{pmatrix} M_4\right](1+\mathcal{O}(1/n))
$$
%%-----
with the diagonal blocks being successively $n_{-}\times n_{-}$, $n_{+}\times n_{+}$ and $(2N - n_{-}- n_{+})\times (2N - n_{-}- n_{+})$ matrices.
$M_3$ and $M_4$ consist in the basis given by $V$ of blocks 
%%-----
$$
   \begin{pmatrix}1 & 1 \\ \mathrm{e}^{c_0^{(s)}+ \bar{a}_{j}} & \mathrm{e}^{-c_0^{(s)} - \bar{a}_{j}}\end{pmatrix}\quad \mathrm{and}\quad
\begin{pmatrix}1 & -1 \\ -\mathrm{e}^{c_0^{(s)}+ \bar{a}_{j}} & \mathrm{e}^{-c_0^{(s)} - \bar{a}_{j}}\end{pmatrix}\,,
$$
%%-----
respectively. Note that there is a zero instead of $D-I$ in the lower-right block of the matrix which multiplies $M_3$, because only the first $n_{-}$
rows of the first matrix contribute to the term by $n^{2N-n_{-}}$. The values of $c_0^{(s)}$ are hence given only by the averages of the damping functions.
The potentials $b_j(x)$ do not play any role in the leading term of the asymptotics of the secular equation, hence they do not influence the
coefficients $c_0^{(s)}$. Using Lemma~\ref{lem-l0} the claim can be generalized for equilateral graphs with lengths $l_0$.
\end{proof}

\begin{remark}\label{remnotcont}{\rm 
Let us remark that unlike the case of a segment with continuous damping, the eigenfunctions of a graph do not in general converge for high $n$
to the eigenfunctions of a quantum graph (the undamped case), as the damping need not be continuous on the whole graph. The simplest counterexample is a segment
of length~2 with constant damping $a_1$ on the first half and 
constant damping $a_2 \ne a_1$ on the second half, with Dirichlet boundary conditions at the end points and standard coupling in the middle. 

Defining $\tilde\lambda_j(\lambda) = \sqrt{\lambda^2+2a_j\lambda-b_j}$ and using the equation (\ref{eq-ldw1}) we find the wavefunction on the 
edges of the graph as
$$
  u_j(x) = \alpha_j \sinh\left(\tilde\lambda_j(\lambda)x\right)\,,\quad j = 1,2
$$
with $x=0$ corresponding to the end vertices ($\cosh\left(\tilde\lambda_j(\lambda)x\right)$ is excluded because of the Dirichlet conditions). Using the 
high-frequency asymptotics of $\tilde\lambda_{j}(\lambda)$ and the coupling conditions at the middle vertex we find the leading term of the
asymptotics, which gives the equation for $c_0^{(s)}$ to be
\begin{multline*}
  0 = \cosh{(a_1+c_0^{(s)})}\sinh{(a_2+c_0^{(s)})}+  \cosh{(a_2+c_0^{(s)})}\sinh{(a_1+c_0^{(s)})} = \\
  = \sinh{(a_1+ a_2+ 2c_0^{(s)})}\,.
\end{multline*}
Hence one obtains the values $c_{0}^{(1)} = -\frac{a_1+a_2}{2}$ and $c_{0}^{(2)} = -\frac{a_1+a_2}{2}+ i$. Let us without loss of generality study 
only the eigenvalues with $c_{0}^{(1)} = -\frac{a_1+a_2}{2}$. For the eigenfunction components we obtain, using 
$\tilde\lambda_j = 2\pi i n+c_0^{(s)}+a_j+\mathcal{O}\left(1/n\right)$,
\begin{multline*}
  u_1(x) \stackrel{n\to \infty}{\longrightarrow} \alpha_1\left[\sin{(2\pi n x)}\cosh{\frac{1}{2}(a_1-a_2)x} -i \cos{(2\pi n x)}\sinh{\frac{1}{2}(a_1-a_2)x}\right]\\
\end{multline*}
\begin{multline*}
  u_2(x) \stackrel{n\to \infty}{\longrightarrow} \alpha_1 \frac{\sinh{\frac{a_1-a_2}{2}}}{\sinh{\frac{a_2-a_1}{2}}}\left[\sin{(2\pi n x)}\cosh{\frac{1}{2}(a_2-a_1)x}-\right.\\
  \left.-i \cos{(2\pi n x)}\sinh{\frac{1}{2}(a_2-a_1)x}\right]
  =  \alpha_1\left[-\sin{(2\pi n x)}\cosh{\frac{1}{2}(a_1-a_2)x} -\right.
  \\
  \left.-i \cos{(2\pi n x)}\sinh{\frac{1}{2}(a_1-a_2)x}\right]
\end{multline*}

On the other hand, in the case where there is no damping we have leading eigenfunction components given by
\begin{eqnarray*}
  v_1(x) \stackrel{n\to \infty}{\longrightarrow} \alpha_1\sin{(2\pi n x)}\,,\\
  v_2(x) \stackrel{n\to \infty}{\longrightarrow} -\alpha_1\sin{(2\pi n x)}\,.
\end{eqnarray*}
Clearly, for $a_1\ne a_2$ the leading eigenfunction components $u_1$ and $u_2$ for the damped case are different from their counterparts for
the undamped case $v_1$ and $v_2$.
}\end{remark}

\begin{remark}{\rm
As a by-product, Theorem~\ref{thm-justif} gives that for the undamped case with $N$
edges of unit lengths with general coupling conditions
there exist at most $2N$ sequences of eigenvalues with $k_n$ satisfying asymptotics~(\ref{eq-asymptot}) with purely imaginary $c_0^{(s)}$.
Here $k_n^2 = E_n$ are the eigenvalues of the Schr\"odinger operator. Note that, as we stated in the Introduction and in
Section~\ref{sec-description}, the operator $i H_{A=0}$ is self-adjoint if all $b_j(x)$
are positive, and hence the eigenvalues of $H_{A=0}$ are on the imaginary axis. Under these conditions, the operator
$i H_{A=0}$ corresponds to an equilateral quantum graph. This generalizes the results of \cite{YY} for star graphs with 
standard coupling and \cite{CP} for general graphs with standard coupling.   
}\end{remark}

\section{Location of high-frequency abscissas}\label{sec-loc}
We begin with a simple result stating that the real part of an eigenvalue is given by minus the average of the damping coefficients on each
edge with weights taken as square absolute values of the eigenfunctions. A similar claim can be stated also for the damped wave equation on manifolds
(see \emph{e.g.} \cite{AL,Sj}).
%%----
\begin{theorem}\label{thm-lambdar}
Let us consider a damped wave equation on a graph with $N$ edges of lengths $l_j$, bounded damping coefficients $a_j(x)$ and potentials
$b_j(x)$, and the coupling conditions given by (\ref{eq-coupling}). 
Then, given a non-real eigenvalue $\lambda$ of the operator $H$ defined in Section~\ref{sec-description}, its real part satisfies
%%----
$$
  \Re(\lambda) =  - \frac{\sum_{j =1}^N  \int_{0}^{l_j} a_j(x) |u_j(x)|^2 \,\mathrm{d}x}{\sum_{j =1}^N \|u_j(x)\|^2_2},
$$
where $u_j(x)$ denotes the corresponding wavefunction components.
%%----
\end{theorem}
\begin{proof}
Multiplying (\ref{eq-ldw1}) by $\bar u_{j}(x)$ and integrating over the $j$-th edge yields
%%----
\begin{multline}
  0 =  \int_{0}^{l_j}\bar u_j(x) u_j''(x) \,\mathrm{d}x - \int_{0}^{l_j} [\lambda^2 + 2 a_j(x)\lambda-b_j(x)] |u_j(x)|^2 \,\mathrm{d}x = 
  \\
  =  - \int_{0}^{l_j} |u_j'(x)|^2 \,\mathrm{d}x + \left[\bar u_j(x) u_j'(x)\right]_0^{l_j} - \int_{0}^{l_j} [\lambda^2 + 2 a_j(x)\lambda-b_j(x)]\,
|u_j(x)|^2 \,\mathrm{d}x\,.\label{eq-lambdar}
\end{multline}
%%----
Now we sum this term over all edges. We show that the contribution of the term given by coupling conditions is real. Let $U$ be the coupling
matrix of the corresponding flower-like graph with $N$ edges of the form (\ref{eq-u})
with $V$ being unitary and denote $V\Psi=(\Psi_1,\Psi_2,\Psi_3)^{\mathrm{T}}$ and by $V\Psi' =(\Psi_1',\Psi_2',\Psi_3')^{\mathrm{T}}$ vectors of
functional values and outward derivatives at
a given vertex written in a suitable basis. Trivially, the subspaces corresponding to eigenvalues $\pm 1$ do not contribute to the second term of the
above equation since either the functional value or the derivative is equal to zero. Hence we have 
%%----
$$
	\sum_{j = 1}^N \left[\bar u_j(x) u_j'(x)\right]_0^{l_j} = \sum_{s = 1}^{2N - n_{-} -n_{+}} \tan{\left(\frac{\varphi_s}{2}\right)}
|\Psi_{3s}|^2
$$
%%----
with $\Psi_{3s}$ being the components of $\Psi_3$. The boundary terms correspond to the boundary terms in a quadratic form
of the Hamiltonian in quantum graphs (see e.g. \cite{Ku1}, Definition 15).

For the imaginary part of (\ref{eq-lambdar}) we have
%%----
\begin{equation}
  0 = 2 \mathrm{i} \,\Im(\lambda)\, \sum_{j = 1}^{N} \int_0^{l_j} (a_j(x) + \Re(\lambda))\,|u_j(x)|^2 \,\mathrm{d}x \label{eq-relambda}
\end{equation}
%%----
and hence the claim of the theorem follows. 
\end{proof}
%%----
\begin{corollary}\label{corbound}
Let us consider a damped wave equation on the graph $\Gamma$ with damping functions on the edges
$a_j(x)$ and potentials $b_j(x)$. 
Then the real part of nonreal eigenvalues of $H$ lie in the interval
\[
 \left[-\sup_j\esssup_{x\in(0,l_j)}\,a_j(x),-\inf_j\essinf_{x\in(0,l_j)}\,a_j(x)\right],
\]
and all high-frequency abscissas lie in the interval $[-\max_{j}{\bar{a}_j},-\min_{j}{\bar{a}_j}]$.
\end{corollary}
\begin{proof}
We have
\[
\begin{array}{lll}
\Re(\lambda) & \geq & - \frac{\ds \sum_{j =1}^N  \esssup_{x\in(0,l_j)}\,a_j(x)  \|u_j(x)\|^2_2}{\ds\sum_{j =1}^N \|u_j(x)\|^2_2}\eqskip
& \geq & -{\ds \sup_j\esssup_{x\in (0,l_j)}}\,a_j(x) \frac{\ds \sum_{j =1}^N    \|u_j(x)\|^2_2}{\ds \sum_{j =1}^N \|u_j(x)\|^2_2}
\end{array}
\]
and similarly for the infimum. The claim about the high-frequency abscissas follows by a similar construction from
Theorems~\ref{thm-lambdar} and~\ref{thm-onlyaverage}.
\end{proof}

For certain types of graphs with standard coupling, some of the high-frequency abscissas have a simple explicit form.

%%-----
\begin{theorem}\label{alla}
Let $\Gamma$ be an equilateral graph with standard coupling which contains a cycle $C$ of $N$ edges of length one with averages of damping
coefficients on all of the edges of $C$ equal to $a$. Then there is a high-frequency abscissa at $-a$. 
\end{theorem}
\begin{proof}
Due to Theorem~\ref{thm-onlyaverage} one can assume all damping coefficients on $C$ equal to $a$ and take $b_j = 0$ without loss of
generality. Our aim is to construct an eigenfunction of $H$ with support only on $C$ and with zeros at the vertices. Using the
Ansatz~$f_j(x) = \alpha_j(\mathrm{e}^{\tilde \lambda(\lambda)x} - \mathrm{e}^{-\tilde \lambda(\lambda)x})$ with
$\tilde\lambda(\lambda) = \sqrt{\lambda^2 + 2 a \lambda}$ on the $j$-th edge one finds
%%-----
\begin{eqnarray*}
  f_j(0) = 0\,,&\quad &f_j(1) = \alpha_j (\mathrm{e}^{\tilde \lambda(\lambda)} - \mathrm{e}^{-\tilde \lambda(\lambda)})\,,\\
  f_j'(0) = 2 \alpha_j \tilde \lambda(\lambda) \,,&\quad &f_j'(1) = \alpha_j \tilde \lambda(\lambda)(\mathrm{e}^{\tilde \lambda(\lambda)} +
\mathrm{e}^{-\tilde \lambda(\lambda)})\,.\\
\end{eqnarray*}
%%-----
Requiring the functional values to be zero at the vertices of the $C$ one finds $\tilde \lambda(\lambda) =  n \pi i$, the continuity of the
derivatives leads to $\alpha_{j+1} -\alpha_j (-1)^n = 0$ and thus $(-1)^{nN} = 0 $. Hence for $N$ even one has two sequences with $c_{0}^{(1)} = - a$
and $c_{0}^{(2)} = - a + \pi i$ while for $N$ odd there is a sequence with $c_{0}^{(1)}= - a$. 
\end{proof}
%%-----

%%-----
\begin{remark} {\rm
A similar claim can be made if a graph contains a chain of edges of the same length and same averages of damping functions with standard coupling in
the middle vertices and the end vertices belonging to the boundary and having \emph{e.g.} Dirichlet coupling. On the other hand, an edge without
damping does not by itself ensure a sequence of eigenvalues with their real parts approaching zero. This can be shown for instance for an
equilateral three-edge star graph with constant dampings on each edge $a$, $a$, and $0$; the leading term of the large $n$ asymptotics of the
corresponding secular equation gives the polynomial equation
\[
   3 \mathrm{e}^{4a} z^3 - (\mathrm{e}^{4a}+2\mathrm{e}^{2a}) z^2 - (2\mathrm{e}^{2a}+1) z +3  = 0
\]
with $z = \mathrm{e}^{2c_0^{(s)}}$. It has solutions $z_1 = \mathrm{e}^{-2a}$, $z_{2,3} = \frac{1}{6} \mathrm{e}^{-2a} (\mathrm{e}^{2a}-1 \pm
\sqrt{\mathrm{e}^{4a}+ 34\, \mathrm{e}^{2a}+1})$ which lead to sequences of eigenvalues with coefficients 
$c_0^{(1)} = -a$, $c_0^{(2,3)} = \frac{1}{2}\ln{z_{2,3}}$.
}\end{remark}
%%-----

%%-----
\begin{remark} {\rm
Theorem~\ref{alla} can be generalized to a larger class of coupling conditions with the matrix $U$ satisfying 
$$
  (U+I) (1,-1,1,-1,\dots,1,-1,0,\dots,0)^\mathrm{T} = 0\,,
$$
where the columns multiplied by $\pm 1$ correspond to the edges of the cycle. For these graphs the term with the derivative in the
coupling conditions disappears for eigenfunctions constructed as in the proof of the Theorem; the term with the functional values
disappears due to the fact that these vanish at the vertices.

For instance, for one loop starting and ending in the same vertex connected to a segment the unitary matrix in the central 
vertex must be of the form
$$
 U = \begin{pmatrix} 
    p & p+1 & \bar q\\
    p+1 & p & \bar q\\
    q & q & r 
    \end{pmatrix}
$$
with $p$, $q$, $r$ being complex and satisfying
$$
  r = -2p-1\,,\quad 2 |p|^2 + |q|^2 + 2\, \mathrm{Re\,}p = 0\,.
$$
Hence we have a 3-parametre set of the whole family of 9-parametre set of unitary matrices $3\times 3$.
From the last equation it also follows that $\mathrm{Re\,}p \in (-1,0)$. The choice $p=-1/3$, $q=2/3$, $r=-1/3$ gives standard conditions.
}
\end{remark}
%%-----

\section{Studying the secular equation via pseudo-orbit expansions}\label{sec-orbit}
In what follows we shall use the approach based on the vertex scattering matrices and its pseudo-orbit
expansion to construct the secular equation.
We refer to the pioneering articles in the case of quantum graphs by Kottos and Smilansky~\cite{KoS} and Akkermans et al.~\cite{Ak1}, the
publication by Band, Harrison and Joyner \cite{BHJ} or the paper by Bolte and Endres where the case of general coupling was worked out \cite{BE}. 

In the case of the damped wave equation similar techniques can be used and, in a similar way to the cited works above, the orbit expansion can
help finding the secular equation. From the secular equation, one can then find the leading term of its high-frequency asymptotics (i.e. the
expansion in $n$). In this section this expansion is done in terms of scattering matrices. Furthermore, we show how it is possible to obtain
particular coefficients in the first term of this expansion using the orbit expansion.

We shall first briefly describe the method of orbit expansion. Let us replace the graph $\Gamma$ by the directed graph $\Gamma_2$, where
the $j$-th edge of $\Gamma$ is replaced by two directed edges (which, following~\cite{BHJ}, we will call \emph{bonds}) $e_j$ and $\hat{e}_j$,
both of length $l_j$. Using the Ansatz
%%-----
\begin{eqnarray*}
   f_{e_j} (x) = \alpha^{\mathrm{in}}_{e_j}\mathrm{e}^{\tilde \lambda_{j} x} + \alpha^{\mathrm{out}}_{e_j}\mathrm{e}^{-\tilde \lambda_{j} x} \,,\\
   f_{\hat{e_j}} (x) = \alpha^{\mathrm{in}}_{\hat{e}_j}\mathrm{e}^{\tilde \lambda_{j} x} +
\alpha^{\mathrm{out}}_{\hat{e}_j}\mathrm{e}^{-\tilde
\lambda_{j} x}
\end{eqnarray*}
%%-----
and the relation $f_{e_j} (x) = f_{\hat{e}_j} (l_j - x)$ one has
%%-----
$$
   \alpha_{\hat{e}_j}^{\mathrm{out}} = \mathrm{e}^{\tilde \lambda_j l_j} \alpha_{e_j}^{\mathrm{in}}\,, \qquad
   \alpha_{e_j}^{\mathrm{out}} = \mathrm{e}^{\tilde \lambda_j l_j} \alpha_{\hat{e}_j}^{\mathrm{in}}\,.
$$
%%-----
The interpretation of the incoming and outgoing waves for this choice cor\-res\-ponds to $\tilde \lambda_j$ as the positive square
root of $\lambda^2 + 2 a_e \lambda - b_e$. It can be compared to the case of quantum graphs, where one has $\mathrm{e}^{ikx}$ 
and $\mathrm{e}^{-ikx}$ with $k^2 = E$ being the spectral parameter. Here $k$ corresponds to $-\Im(\lambda)$.
We define the vertex scattering matrix $\sigma_v (\lambda)$ by
$\vec{\alpha}_v^{\mathrm{out}} = \sigma_v (\lambda)\vec{\alpha}_v^{\mathrm{in}}$ with $\vec{\alpha}_v^{\mathrm{in}} =
(\alpha_{e_{v1}}^{\mathrm{in}},\dots,\alpha_{e_{vd}}^{\mathrm{in}})^\mathrm{T}$ and $\vec{\alpha}_v^{\mathrm{out}} =
(\alpha_{e_{v1}}^{\mathrm{out}},\dots,\alpha_{e_{vd}}^{\mathrm{out}})^\mathrm{T}$. 

The matrix $\Sigma(\lambda)$ is defined by the equation
$$
  \vec{\alpha}^{\mathrm{out}} = \Sigma(\lambda)\vec{\alpha}^{\mathrm{in}}\,,
$$
where
\begin{eqnarray*}
\vec{\alpha}^{\mathrm{out}} = (\alpha_{e_1}^{\mathrm{out}},\dots, \alpha_{e_N}^{\mathrm{out}},\alpha_{\hat{e}_1}^{\mathrm{out}},\dots,
\alpha_{\hat{e}_N}^{\mathrm{out}})^\mathrm{T} \,,
\\
\vec{\alpha}^{\mathrm{in}} = (\alpha_{e_1}^{\mathrm{in}},\dots, \alpha_{e_N}^{\mathrm{in}},\alpha_{\hat{e}_1}^{\mathrm{in}},\dots,
\alpha_{\hat{e}_N}^{\mathrm{in}})^\mathrm{T}\,.
\end{eqnarray*}
The matrix $\Sigma(\lambda)$ is block-diagonalizable via the similarity transformation $W\Sigma(\lambda)W^{-1}$. Here
the matrix $W$ maps the vector $\vec{\alpha}^{\mathrm{in}}$ to the vector
$$
  \left[(\vec{\alpha}_{v_1}^{\mathrm{in}})^\mathrm{T},(\vec{\alpha}_{v_2}^{\mathrm{in}})^\mathrm{T}, \dots,
  (\vec{\alpha}_{v_{|\mathcal{V}|}}^{\mathrm{in}})^\mathrm{T}\right]^\mathrm{T}\,,
$$ 
where $|\mathcal{V}|$ is the number of vertices in the graph. This similarity transformation only rearranges rows and columns of $\Sigma(\lambda)$
and the blocks of the matrix $W\Sigma(\lambda)W^{-1}$ are $\sigma_{v_j}(\lambda)$.

Finally, we define 
$$
 J = \begin{pmatrix}0& I\\I & 0\end{pmatrix}\quad\mathrm{and}\quad L =\mathrm{exp}\left(\mathrm{diag\,}(-\tilde \lambda_1 l_1,\dots, -\tilde \lambda_N
l_N,-\tilde \lambda_1 l_1,\dots, -\tilde \lambda_N l_N)\right)
$$ 
and obtain
%%-----
$$
  \begin{pmatrix}\vec{\alpha}_e^{\mathrm{in}}\\\vec{\alpha}_{\hat{e}}^{\mathrm{in}} \end{pmatrix} = L
\begin{pmatrix}\vec{\alpha}_{\hat{e}}^{\mathrm{out}}\\\vec{\alpha}_{e}^{\mathrm{out}} \end{pmatrix} = L J
\begin{pmatrix}\vec{\alpha}_{e}^{\mathrm{out}}\\\vec{\alpha}_{\hat{e}}^{\mathrm{out}} \end{pmatrix} =  L J
\Sigma(\lambda)\begin{pmatrix}\vec{\alpha}_e^{\mathrm{in}}\\\vec{\alpha}_{\hat{e}}^{\mathrm{in}} \end{pmatrix} 
$$
%$$
%   \vec{\alpha}^{\mathrm{in}} = L J \Sigma(\lambda) \vec{\alpha}^{\mathrm{in}}\,.
%$$
%%-----
Hence the secular equation is
%%-----
\begin{equation}
   \mathrm{det\,}\left(L J \Sigma(\lambda) - I_{2N\times 2N} \right) = 0\,.\label{eq-seceq}
\end{equation}
%%-----
We further denote $S(\lambda) = J \Sigma(\lambda)$.

The following proposition compares the vertex scattering matrix behaviour in $1/n$ expansion with the coupling matrix for this particular vertex.
%%-----
\begin{proposition}(Vertex scattering matrix behaviour)\label{thm-scatmat}\\
Let us assume a vertex with coupling given by a unitary matrix $U$.
%\begin{enumerate}
%\item[i)] if spectrum $\sigma(U) = \{-1,1\}$ then $\sigma(\lambda) = U$,
%\item[ii)] 
If $U$ is of the form (\ref{eq-u}), then
\[
\sigma(\lambda_{sn}) = V^{-1}\begin{pmatrix} -I_{n_{-}} & 0 & 0\\ 0 & I_{n_{+}} & 0 \\ 0 & 0 & I_{d-n_{-}-n_{+}}\end{pmatrix} V + \mathcal{O}(1/n) 
\] 
where $d$ is the vertex degree.
%\end{enumerate}
\end{proposition}
\begin{proof}
The coupling condition becomes 
%%-----
$$
  (U -I) (\vec{\alpha}^{\mathrm{in}}+\vec{\alpha}^{\mathrm{out}}) + i (U+ I)(2\pi i n I + \Lambda_2)
(\vec{\alpha}^{\mathrm{in}}-\vec{\alpha}^{\mathrm{out}}) = 0\,,
$$
%%-----
with $\Lambda_2$ being the diagonal matrix with entries $c_0^{(s)}+ a_j + \mathcal{O}(1/n)$ and with $U$ being the coupling matrix 
defined in Section~\ref{sec-description}. This is due to the high-frequency asymptotics of 
$\tilde \lambda_j(\lambda)$. The previous equation leads to the relation
%%-----
\begin{equation}
   [(U-I)+ 2 \pi n (U+I) - i (U+I)\Lambda_2] \sigma(\lambda_{sn}) = - [(U-I)-2 \pi n (U+I) + i (U+I) \Lambda_2] \label{eq-sigma}.
\end{equation}
%%-----
Using the expansion of $\sigma(\lambda) = \sigma_0 + \frac{1}{n}\sigma_1 + \frac{1}{n^2} \sigma_2+ \mathcal{O}(1/n^3)$ in $n$
one finds from the first two terms
%%-----
\begin{eqnarray}
  2 \pi n (U+I) (\sigma_0 - I) = 0\,,\label{eq-sigma1}\\
  (U -I) (\sigma_0 + I) + 2 \pi (U+I) \sigma_1 + i (U+I) \Lambda_2 (I - \sigma_0) = 0\,.\label{eq-sigma2}
\end{eqnarray}
%%-----
For simplicity, we study the whole problem in the basis in which $U$ is diagonal. Equation~(\ref{eq-sigma1}) implies the following form 
 for the constant (in $\lambda_{sn}$) part of $\sigma$
%%-----
$$
  \sigma_0 = \begin{pmatrix}\sigma_0^1 & \sigma_0^2 & \sigma_0^3 \\ 0 & I & 0\\ 0 & 0 & I  \end{pmatrix}\,.
$$
%%-----
The fact that $\sigma_0$ is unitary yields $\sigma_0^2 = \sigma_0^3 = 0$ and since all upper blocks of the second and third terms in
(\ref{eq-sigma2}) vanish, we have $\sigma_0^1 = -I_{n_{-}}$. Rewriting $\sigma(\lambda)$ into the appropriate
basis we obtain the claim.
%\end{enumerate}
\end{proof}
%%-----

We use the same terminology as in \cite{BHJ}.
\begin{definition}
A \emph{periodic orbit} is a closed trajectory on the graph $\Gamma_2$. An \emph{irreducible pseudo orbit} $\bar\gamma$ is a collection of
periodic orbits where 
none of the directed bonds is contained more than once. Let $m_{\bar{\gamma}}$ denote the number of periodic orbits in $\bar{\gamma}$,
$L_{\bar{\gamma}} = \sum_{e \in \bar{\gamma}}\tilde\lambda_e l_e$ where the sum is over all directed bonds in $\bar{\gamma}$ and
$\tilde\lambda_e = \sqrt{\lambda^2 + 2 a_e \lambda - b_e}$. The coefficients
$A_{\bar{\gamma}} = \Pi_{\gamma_j \in \bar{\gamma}} A_{\gamma_j}$ with $A_{\gamma_j}$ given as multiplication of entries of $S (\lambda)$ along the
trajectory $\gamma_j$. 
\end{definition}
%%-----

\begin{theorem}
The secular equation for the damped wave equation on a metric graph is given by 
$$
  \sum_{\bar \gamma} (-1)^{m_{\bar{\gamma}}} \, A_{\bar{\gamma}}(\lambda) \,\mathrm{exp}(-L_{\bar{\gamma}}(\lambda))= 0\
$$
with $L_{\bar{\gamma}}$ being the sum of the lengths of all directed bonds along a particular irreducible pseudo orbit $\bar{\gamma}$.
\end{theorem}
\begin{proof}
The proof is similar to that of Theorem~1 in~\cite{BHJ}, except that the role of $\mathrm{exp}(i k l_{\bar{\gamma}})$ is replaced by
$\mathrm{exp}(-L_{\bar{\gamma}}(\lambda))$. Its main essence is that nontrivial contributions to the determinant of (\ref{eq-seceq}) are given by such
permutations which correspond to irreducible pseudo orbits. Each pseudo orbit of length $s$ corresponds to taking $s$ entries of the matrix
$S(\lambda)$ and $2N-s$ entries of $-I$. Since the number of $-1$'s coming from the entries of this last matrix cancels with a part of the
sign contribution of the permutation coming from the first term, the sign of a given contribution is given by the number of periodic orbits
in~$\bar{\gamma}$.
\end{proof}
%%-----

%%-----
\begin{theorem}\label{thm-damprate}
Let $\Gamma$ be a graph with given damping coefficients $a_j(x)$ and potentials $b_j(x)$. At a fixed vertex $v$ we assume a general coupling
given by the matrix~$U$.
\begin{enumerate}
\item[i)] If $-1 \not \in \sigma(U)$ then the abscissa count and high-frequency abscissas will not change if $U$ is replaced by $U_\mathrm{N} = I$,
\emph{i.e.} we separate all edges with Neumann coupling.
\item[ii)] If there is $\delta$-coupling of strength $\alpha\in \mathbb{R}$ ($U = \frac{2}{d+ i\alpha}J -I$) then the abscissa count
and high-frequency
abscissas will not change if $U$ is replaced by standard coupling ($\alpha = 0$). We emphasize that the case $\alpha = \infty$, \emph{i.e.} fully
separated case, is not included. 
\end{enumerate}
\end{theorem}
\begin{proof}
\begin{enumerate}
\item[i)] According to Proposition~\ref{thm-scatmat} the first term of the asymptotics of the vertex scattering matrix equals
the identity matrix. Hence the first term of the secular equation is the same as in
the case of Neumann coupling. It is straightforward to verify that the
$\delta'_\mathrm{s}$-coupling matrix does not have $-1$ in its spectrum.
\item[ii)] The $\delta$-coupling matrix has one eigenvalue equal to $\frac{d-i\alpha}{d+i\alpha}$ and an eigenvalue $-1$ with multiplicity
$d-1$. Hence, according to proposition~\ref{thm-scatmat}, the first term in the expansion of the vertex scattering matrix is the same as 
that in the expansion of the vertex scattering matrix for the case of the standard coupling.
\end{enumerate}
\end{proof}
%%-----
The cases considered in i) include \emph{e.g.} $\delta'_\mathrm{s}$-condition ($U = I -
\frac{2}{d-i \beta}$) with strength $\beta\in \mathbb{R}\backslash\{0\}$.

\section{Number of distinct high-frequency abscissas}\label{sec-number}
Our next aim is to bound the number of sequences of eigenvalues corresponding to different 
high-frequency abscissas. We begin by stating the main theorem of this section, whose proof is based on several lemmata and given
at the end of the section. 

To describe the distribution of eigenvalues and compare it with the two-dimensional results we define the following probability measure (see also
\emph{e.g.} \cite{AL,Sj}).
%%-----
\begin{definition}
Let $I$ be an open interval in $\mathbb{R}$ and $R >0$. Then we define the probability distribution $\mu_R(I)$ by
%%-----
$$
  \mu_R(I) = \frac{\#\{\lambda: \mathrm{Re\,}\lambda \in I, |\mathrm{Im\,}\lambda|< R\}}{\#\{\lambda: |\mathrm{Im\,}\lambda|< R\}}\,.
$$
%%-----
We define $\mu_\infty (I)$ by ${\ds\lim_{R\to \infty}} \mu_R(I)$.
\end{definition}
%%-----

We define a bipartite graph and state its main property. For more details see e.g.~\cite{Di}.
\begin{definition}
A graph is \emph{bipartite} if it admits partition into two classes such that every edge ends in different classes.
\end{definition}

\begin{proposition}
A graph is bipartite iff it does not have any closed cycle of odd length.
\end{proposition}

The main result of the paper can be written in terms of $\mu_\infty(I)$ in the following way.
%%-----
\begin{theorem}\label{thm-main}
Let $\Gamma$ be an equilateral graph with $N$ edges of length one, with coupling given by~(\ref{eq-coupling}) and with damping and
potential functions $a_j(x)\in C^1[0,1]$, $b_j(x)\in C^0[0,1]$, possibly discontinuous at the vertices. 
\begin{enumerate}
\item[i)] The measure $\mu_\infty$ is atomic with atoms with measures given by $\frac{m_i}{2N}$ with $m_i$ being a positive
integer. The number of atoms is at most $2N$ and ${\ds \sum_{i=1}^{\alpha_c} m_{i} = 2N}$.
\item[ii)] If the graph is bipartite with Robin coupling on the boundary and standard coupling otherwise, then all $m_{i}$'s are even.
\item[iii)] For a tree graph with Robin coupling on the boundary and standard coupling otherwise, having all vertices of odd
degree, there always exists a damping for which the maximum possible number of~$N$ atoms allowed by i) and ii) is
achieved.
\end{enumerate}
\end{theorem}

\begin{remark}\label{rem-polyeq}{\rm
We present a systematic procedure to determine the polynomial whose solutions yield the high-frequency abscissas.
\begin{enumerate}
\item Given a graph with commensurate edges, introduce new vertices of degree two to make the graph equilateral;
\item For an equilateral graph with edges of length $l_0$ use the construction of Lemma~\ref{lem-l0} to obtain an equilateral graph
with unit lengths;
\item Replace the damping on each edge by its average and make all potentials equal to zero (using Theorem~\ref{thm-onlyaverage});
\item Construct a secular equation by first constructing two independent solutions $\sinh{(\tilde \lambda_j(\lambda) x)}$
and $\cosh{(\tilde \lambda_j(\lambda) x)}$ on each edge and then use the coupling conditions~(\ref{eq-coupling});
\item Determine the leading term of the secular equation and replace $\tilde\lambda_j(\lambda)$ by $c_0^{(s)}+\bar a_j$ in the arguments
of the hyperbolic functions;
\item Determine the roots $z_i$ of the corresponding polynomial equation in $z = \mathrm{e}^{c_0^{(s)}}$ and $\omega_0$, the real parts of $c_0^{(s)}$, as
$\ln{|z_i|}$.
\end{enumerate}
}\end{remark}

For the proof of the Theorem~\ref{thm-main} we will use the following lemmata. First, we give an upper bound on the number of these sequences for
an equilateral graph with general coupling conditions. 
%%-----
\begin{lemma}\label{thm-seq1}
Let $\Gamma$ be an equilateral graph with $N$ edges of the length $1$. Let us assume a damped wave equation on $\Gamma$ with damping and potential
functions constant on each edge $a_j(x)\equiv a_j$, $b_j(x) \equiv b_j$ and with general coupling given by~(\ref{eq-coupling}) for
a given unitary matrix $U$. Then
there exist numbers $n_0 \in \mathbb{N}$, $c_0^{(s)}\in \mathbb{C}$, $s = 1, \dots, 2N$ and $c_1 \in \mathbb{R}$ such that for every $n \geq n_0$ all
eigenvalues of $H$ are within the following set
$$
  \{\lambda, |\mathrm{Im\,}\lambda| \leq 2 \pi n_0\}\cup \bigcup_{s=1}^{2N} \bigcup_{n = n_0}^{\infty} B\left(2\pi n i
+c_{0}^{(s)},\frac{c_1}{n}\right)\,,
$$
where $B(x_0,r)$ denotes the disc in the complex plane with center $x_0$ and radius~$r$.
\end{lemma} 
\begin{proof}
As follows from the construction in Section~\ref{sec-seceq}, high $n$ asymptotics of the eigenvalues are given by equation~(\ref{eq-asymptot}).
We choose $c_1$ as twice the maximum of modul\ae\ of the coefficients by the $\mathcal{O}(1/n)$ for particular
sequences and $n_0$ large enough for the $\mathcal{O}(1/n^2)$ term to be smaller than $c_1/n$. It remains to prove that there are at most $2N$
different coefficients $c_0^{(s)}$. Since the first term of the $n$ asymptotics of the secular equation uniquely determines the coefficients 
$c_0^{(s)}$, one can
use the pseudo-orbit expansion with the first term of $n$ asymptotics of vertex scattering matrices. Each pseudo orbit $\bar{\gamma}$ of length
$l_{\bar{\gamma}}$ corresponds to a term by $y^{l_{\bar{\gamma}}}$. Since the directed graph corresponding to $\Gamma$ has $2N$ edges, we obtain a
polynomial equation in $y = \mathrm{e}^{c_0^{(s)}}$ of $2N$-th order, which has $2N$ roots uniquely determining the values of $c_0^{(s)}$.
\end{proof}
%%-----

%%-----
We shall now consider some particular cases for which it is possible to be more specific regarding the possible number of high-frequency
abscissas. In the first example, we show that for bipartite graphs with $N$ edges there are at most $N$ spectral abscissas.
\begin{lemma}\label{thm-nseq}
Let $\Gamma$ be an equilateral graph with $N$ edges of unit length, (general) Robin coupling at the boundary and standard coupling otherwise.
Let us suppose that the graph is bipartite. Then for any damping functions bounded
and $\mathcal{C}^1$ on each edge there are at most $N$ high-frequency abscissas.
\end{lemma}
\begin{proof}
With the construction of the previous section in mind, one can see that all closed orbits have an even number of edges. Hence in the first term of the
$n$ asymptotics of the secular equation there are only terms with $\mathrm{e}^{2 c_0^{(s)}}$. The corresponding polynomial equation in
$\mathrm{e}^{2c_0^{(s)}}$
has only $N$ roots which give at most $N$ different values of $\omega_0 =\mathrm{Re\,}c_0^{(s)}$. 
\end{proof}
%%-----
To obtain a lower bound on the number of high-frequency abscissas we shall first prove the following technical lemma.
\begin{lemma}\label{lem-combinat}
Let $\Gamma$ be a tree graph with standard coupling and all edges of length one, $\Gamma_2$ the corresponding oriented graph. Let $\{e_1,\dots,
e_{2N}\}$ be a set of edges on $\Gamma$, $\bar \gamma$ a pseudo orbit on $\{ e_1,\dots, e_{2N}, \hat{e_1},\dots, \hat{e_{2N}}\}\subset \Gamma_2$ and
$\mathcal{X}$ be a vertex of $\Gamma$ of degree $d$ and let there be $v$ edges emanating from $\mathcal{X}$ denoted by $\{e_1,\dots, e_v\}$. Let $s_1
= \frac{2}{d} - 1$, $s_2 = \frac{2}{d}$ be on-diagonal and off-diagonal elements of the scattering matrix at $\mathcal{X}$, respectively. For a
particular pseudo orbit $\bar\gamma' $ let $\Gamma_3(\bar\gamma')$ be a collection of all pseudo orbits which can be obtained from $\bar\gamma'$ by
all possible changes at $\mathcal{X}$. Then the coefficient in $\sum_{\bar\gamma \in\Gamma_3(\bar\gamma')} A_{\bar\gamma}(\lambda)$ corresponding to
the vertex $\mathcal{X}$ is $A_{\mathcal{X}} = s_1^{v} (s-1)^{v-1} [(v-1)s+1]$ with $s = \frac{s_2}{s_1} = \frac{2}{2-d}$. 
\end{lemma}
\begin{proof}
First, we prove by induction that if one sums up the contribution of all paths with no reflection coefficients at $\mathcal{X}$ one obtains
$(1-v)s_2^{v}$. Let us assume a star graph of $v$ edges and its corresponding directed counterpart. Let us denote by $g(v)$ the
$(-1)^{m_{\bar\gamma}}$ multiple of the number of different pseudo orbits $\bar\gamma$ on it which do not have any reflection at the central vertex
and let us assume $g(v) = 1-v$. This clearly holds true for $v = 2, 3$, since $g(2) = -1$ (pseudo orbit $e_1 \hat{e_2} e_2 \hat{e_1}$) and $g(3) = -2$
(pseudo orbits $e_1 \hat{e_2} e_2 \hat{e_3} e_3 \hat{e_1}$ and $e_1 \hat{e_3} e_3 \hat{e_2} e_2 \hat{e_1}$). Furthermore, we prove that $g(v+1) = v
g(v) - v g(v-1)$. One considers the set of all pseudo orbits without reflection on $v$ edges. The first term corresponds to adding the $(v+1)$-th edge
to one of the pseudo orbits (for each pseudo orbit on $r$ edges one has $r$ possibilities and hence $g(v)$ is multiplied by $v$). The second term
corresponds 
to a pseudo orbit on two edges -- the $(v+1)$-th one and one of the others and any possible combination of nonreflection pseudo orbits on the other
edges, \emph{i.e.} $g(v-1)$. By induction we have $g(v+1) = v (1-v) - v (1-v+1) = -v$.

As the second step we use this result to prove the lemma. Then the coefficient obtained from the pseudo-orbit expansion by $\mathcal{X}$ is given by 
\begin{multline*}
A_{\mathcal{X}} = \sum_{n=0}^v {v \choose n} (-s_1)^n s_2^{v-n} (1- v + n) = (-s_1)^v  \sum_{n = 0}^v (1-n) {v \choose n} (-s)^n = 
\\
=-(-s_1)^v \sum_{n = 0}^v \left[{v-1 \choose n-1} (v-1) - {v-1 \choose n} \right] (-s)^n =
\\
= s_1^v \sum_{n=0}^{v-1} {v-1 \choose n} s^n (-1)^{v-n} [(v-1)s+1] = - s_1^v (s-1)^{v-1} [(v-1)s+1]\,,
\end{multline*}
where in the middle term the expressions ${v-1\choose -1}$ and ${v-1\choose v}$ are taken to be zero.
\end{proof}
%%-----

We may now show that under certain conditions it is always possible to have at least a number of high-frequency abscissas
which is equal to the number of edges in the graph.
%-----
\begin{lemma}\label{thm-below}
Let $\Gamma$ be a tree graph with $N$ edges all with unit length, Robin coupling at the boundary and standard coupling otherwise. Let us suppose
that all vertices have odd degree. Then there always exists such a damping for which the number of high-frequency abscissas is greater than or equal
to $N$. 
\end{lemma}
\begin{proof} 
Let us assume a given set of the edges of the graph $\Gamma$ and all pseudo orbits which go through every edge of this set twice. Then the
contribution of all these pseudo orbits cancels (see \emph{e.g} \cite{Ak1}) if and only if there is a vertex of $\Gamma$ of degree $d= 2v$ and the
above set of edges contains $v$ edges emanating from this vertex. This follows from the previous lemma, since $d = 2 - \frac{2}{s} = 2+ 2(v-1) = 2v$
for $s = -\frac{1}{v-1}$. 

We will prove the lemma by explicitly constructing the damping function for which this maximum of the number of sequences is attained. The idea is to
choose the damping in a way that its average on each edge differs significantly. In the virtue of theorem \ref{thm-onlyaverage} we choose constant
damping on each edge, \emph{i.e.} $0 \ll a_N \ll a_{N-1} \ll \dots \ll a_1$. The first term of the $n$ expansion of the secular equation can be
written as (for simplicity we omit $n$ to the corresponding power). 
%-----
\begin{multline*}
  C_N \mathrm{e}^{2 a_1+ 2 a_2+ \dots + 2a_N} y^{N} + C_{N-1} \mathrm{e}^{2 a_1+ 2 a_2+ \dots + 2 a_{N-1}}\left[1+\mathcal{O}\left(\mathrm{e}^{-2(
a_{N-1}-a_{N})}\right)\right] y^{N-1} +
  \\+ \dots + C_{2}\mathrm{e}^{2 a_1+ 2 a_2}\left[1+ \mathcal{O}\left(\mathrm{e}^{-2(a_{2}-a_{3})}\right)\right] y^{2} + C_{1}\mathrm{e}^{2 a_1}\left[
1+\mathcal{O}\left(\mathrm{e}^{-2(a_{1}-a_{2})}\right)\right] y + C_0 = 0\,,
\end{multline*}
%-----
where $y = \mathrm{e}^{2c_0^{(s)}}$ and the coefficients $C_j$ are given by the orbit expansion. Since there are no vertices of degree two and standard
conditions are considered, none of them is trivial. For $y$ being close to $\mathrm{e}^{-2a_1}$ the last two terms are dominant, for $y$ close to
$\mathrm{e}^{-2a_2}$ the last-but-one and last-but-two terms, etc. Hence for the roots of the previous polynomial equation of the $N$-th order we get
%%%-----
\begin{equation}
  y_j =  -\frac{C_{j-1}}{C_j}\,\mathrm{e}^{-2a_j}\left[ 1+\mathcal{O}\left(\mathrm{e}^{-2(a_{j}-a_{j+1})}\right)\right]\,,\label{eq-difdampsol}
\end{equation}
%-----
which gives 
%-----
$$
  c_0^{(j)} =  - a_j + \frac{1}{2}\ln{\left(-\frac{C_{j-1}}{C_j}\right)}+ \ln{\left[1+\mathcal{O}\left(\mathrm{e}^{-2(a_{j}-a_{j+1})}\right)\right]}
\,.
$$
%-----
The above result can be generalized for all equilateral graphs without vertices of degree two. The proof is very similar. Let us choose the damping
coefficients again as $0 \ll a_N \ll a_{N-1} \ll \dots \ll a_1$. Then the difference of the term of the pseudo-orbit expansion of the secular equation
 corresponding to $2j+ 1$ edges and the term corresponding to $2j$ or $2j+2$ edges is $\mathcal{O}(\mathrm{e}^{-a_j+a_{j+1}+a_{j+2}})$, since the term
with odd number of edges must contain at least three edges only once. Hence the secular equation for a graph with cycles can be viewed as a small
perturbation of a tree graph for this choice of the damping. Each of the roots~(\ref{eq-difdampsol}) gives rise to two roots (in general not with
distinct real parts) 
$$
  z_{j1,2} =  \sqrt{-\frac{C_{j-1}}{C_j}}\,\mathrm{e}^{-a_j}\left[
1+\mathcal{O}\left(\mathrm{e}^{-2(a_{j}-a_{j+1})}\right)+\mathcal{O}\left(\mathrm{e}^{-(a_{j}-a_{j+1}-a_{j+2})}\right)\right]\,,
$$
(with $z_j = \mathrm{e}^{c_0^{(j)}}$) which are still far away from the other roots. 
\end{proof}
%-----

\begin{proof}[Proof of Theorem~\ref{thm-main}]\ \\
\begin{enumerate}
\item[i)] The claim follows from Lemma~\ref{thm-seq1}. The number $m_{i}$ is an integer since the difference 
between the imaginary parts of two consecutive eigenvalues in
each sequence $\lambda_{sn}$ is asymptotically $2\pi$. Hence each $c_0^{(s)}$ corresponding to one root of polynomial equation of the $2N$-th order
gives a sequence of eigenvalues with the counting function $N_j(R) = R/2\pi+\mathcal{O}(1)$ and therefore
\[
  \lim_{R\to\infty}\mu_R((\mathrm{Re\,}c_0^{(s)}-\varepsilon,\mathrm{Re\,}c_0^{(s)}+\varepsilon)) = \lim_{R\to\infty} \frac{\frac{m_{i}
R}{2\pi}+\mathcal{O}(1)}{\frac{NR}{\pi}+\mathcal{O}(1)} = \frac{m_{i}}{2N}\, ,
\]
proving i).   
\item[ii)] The claim follows from Lemma~\ref{thm-nseq}.
\item[iii)] The claim follows from Lemma~\ref{thm-below}.
\end{enumerate}
\end{proof}
%%-----

\section{Examples}\label{sec-examples}
Now, we present several simple examples to illustrate the asymptotic behaviour of high-frequency eigenvalues.
\begin{example}\label{ex-twocycles}{\rm \textbf{Two cycles with different damping coefficients}\\ 
Let us study an example of a graph with six edges of length one consisting of two cycles joined at one vertex and thus 
not having any boundary vertices (see figure \ref{fig1}). Direct consideration of these six edges would yield a $12\times 12$
coupling matrix. Since in the vertices connecting only two edges both the functional value and the derivative
are continuous, these vertices may be deleted to obtain a graph with two cycles joined at the central vertex. We thus obtain
the $4\times 4$ coupling matrix
$$
   U = \frac{1}{2}\begin{pmatrix}
    -1 & \phantom{-}1 & \phantom{-}1 & \phantom{-}1\\
    \phantom{-}1 & -1 & \phantom{-}1 & \phantom{-}1\\
    \phantom{-}1 & \phantom{-}1 & -1 & \phantom{-}1\\
    \phantom{-}1 & \phantom{-}1 & \phantom{-}1 & -1
   \end{pmatrix}\,.
$$
This corresponds to a coupling matrix at the middle vertex for standard coupling connecting four edges.

Let the damping coefficients $a_1$ and $a_2$ be different on each cycle and let us assume standard coupling at each
of the vertices.
We use the Ansatz $f_j(x) = \alpha_j \sinh{(\tilde \lambda_j(\lambda) x)} + \beta_j \cosh{(\tilde \lambda_j(\lambda) x)}$ on both cycles of the graph
with $x=0$ at their centers. Then from continuity conditions in the central vertex one has
%%----
\begin{eqnarray*}
  \alpha_j \sinh{\left(\frac{3}{2}\tilde \lambda_j(\lambda)\right)} + \beta_j \cosh{\left(\frac{3}{2}\tilde \lambda_j(\lambda)\right)} =  - \alpha_j
\sinh{\left(\frac{3}{2}\tilde \lambda_j(\lambda)\right)} + \beta_j \cosh{\left(\frac{3}{2}\tilde \lambda_j(\lambda)\right)}
\end{eqnarray*}
%%----
and hence either $\alpha_1 = \alpha_2 = 0$, or $\sinh{\left(\frac{3}{2}\tilde \lambda_1(\lambda)\right)} = 0$ or $\sinh{\left(\frac{3}{2}\tilde
\lambda_2(\lambda)\right)} = 0$. Under the first assumption ($\alpha_1 = \alpha_2 = 0$) one has for standard conditions
%%----
\begin{eqnarray*}
  \beta_1 \cosh{\frac{3 \tilde \lambda_1(\lambda)}{2}} = \beta_2 \cosh{\frac{3 \tilde \lambda_2(\lambda)}{2}}\,,\\
  \beta_1  \tilde \lambda_1(\lambda) \sinh{\frac{3\tilde \lambda_1(\lambda)}{2}} + \beta_2 \tilde \lambda_2 (\lambda)
\sinh{\frac{3\tilde\lambda_2(\lambda)}{2}} = 0\,.
\end{eqnarray*}
%%----
with
%%----
$$
  \tilde \lambda_j \equiv \tilde \lambda_j(\lambda) = \sqrt{\lambda^2 + 2 a_j \lambda - b_j}\,.
$$
%%----
Hence 
%%----
$$
  \begin{pmatrix}
   \cosh{\frac{3 \tilde \lambda_1}{2}}&-\cosh{\frac{3 \tilde \lambda_2}{2}} \\
   \tilde \lambda_1 \sinh{\frac{3 \tilde \lambda_1}{2}}& \tilde \lambda_2\sinh{\frac{3 \tilde \lambda_2}{2}} \\
  \end{pmatrix}
  \begin{pmatrix}
    \beta_1 \\ \beta_2
  \end{pmatrix} = 0
$$
%%----
and
%%----
$$
  \tilde \lambda_2 \sinh{\frac{3\tilde\lambda_2}{2}}\cosh{\frac{3\tilde\lambda_1}{2}} + \tilde \lambda_1
\sinh{\frac{3\tilde\lambda_1}{2}}\cosh{\frac{3\tilde\lambda_2}{2}} = 0 \,.
$$
%%----
which can be written as
%%----
$$
  (\tilde \lambda_1 + \tilde \lambda_2) \sinh{\frac{3(\tilde \lambda_1 + \tilde \lambda_2)}{2}} + (\tilde \lambda_1 - \tilde \lambda_2)
\sinh{\frac{3(\tilde \lambda_1 - \tilde \lambda_2)}{2}} = 0\,.
$$
%%----
Using the asymptotic expansion (\ref{eq-asymptot}) one obtains 
%%----
$$
  4 \pi i n \left[\mathrm{e}^{6\pi i  n+ \frac{3}{2}(a_1+a_2+ 2c_0^{(s)})} - \mathrm{e}^{-6 \pi i n -\frac{3}{2}(a_1+a_2+ 2c_0^{(s)})} \right]
  +\mathcal{O}(1) = 0\,. 
$$
%%----
Therefore, from the leading term of the asymptotics one obtains a polynomial equation
$$
  \mathrm{e}^{3(a_1+a_2)}z^6-1 = 0
$$
in $z = \mathrm{e}^{c_{0}^{(s)}}$, from which one can find the coefficients $c_0^{(s)}$ and hence the high-frequency abscissas. We have
$$
  z= \mathrm{exp\,}\left(-\frac{a_1+a_2}{2}+\frac{\pi \mathrm{i}s}{3}\right)
$$
and hence 
%%----
\begin{equation}
  c_0^{(s)} = -\frac{a_1+ a_2}{2}+ \frac{\pi i s}{3}\,,\quad s \in \{0,\dots, 5\} \label{eq-c01}
\end{equation}
%%----
On the other hand, the equation $\sinh{\left(\frac{3}{2} \tilde \lambda_j (\lambda_{sn})\right)} = 0$ gives, in the leading term 
of the asymptotics, the equation
\begin{equation}
  \mathrm{e}^{3a_j} z^{3} -1 = 0 \label{eq-ex1-2} 
\end{equation}
with $z = \mathrm{e}^{c_0^{(s)}}$. Hence
%%----
$$
    c_{0}^{s} = - a_j +\frac{2\pi i s}{3}\,,\quad s \in
\{0, 1, 2\}
$$
%%----
The location of the eigenvalues for a particular choice of the damping coefficients $a_1$ and $a_2$ is shown in figure \ref{ex71_1}, there are in
general three sequences of eigenvalues with real parts given by the previous relation and (\ref{eq-c01}): $-a_1$, $-a_2$ and $-\frac{a_1+a_2}{2}$. The
eigenfunctions corresponding to the first two of them are supported at each of the cycles, the third one is supported on both of them and is symmetric
on each of the cycles. Since one can consider only two edges of length three, the upper bound on the number of distinct high-frequency abscissas given 
by Lemma~\ref{thm-seq1} is four.
%%-----
\begin{figure}[h]
   \begin{center}
   \includegraphics{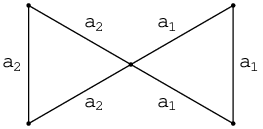}
     \caption{Graph with two cycles (Example \ref{ex-twocycles})}
     \label{fig1} 
   \end{center}
\end{figure} 
\begin{figure}[h]
   \begin{center}
   \includegraphics[width=\textwidth]{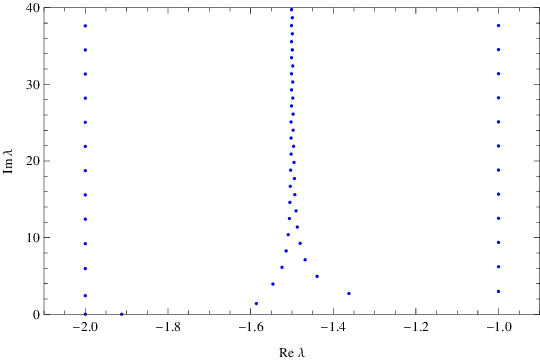}
   \caption{Spectrum of the graph in figure \ref{fig1} for $a_1 = 2$, $a_2 = 1$, $b_1 = 0$, $b_2 = 0$ in Example \ref{ex-twocycles}}\label{ex71_1}
   \end{center}
\end{figure} 
%%-----
}\end{example}

%%-----
\begin{example}\label{ex-loop1ap}{\rm \textbf{One cycle with one appended edge}\\
The graph in the second example consists of a cycle of length three and one edge of length one attached to this cycle (see figure \ref{fig2}).
We assume Dirichlet conditions at the boundary vertex and standard coupling at the central one.
As in the previous example, it would also be possible to replace the three edges in the cycle by one single loop, but in this
case we did not feel the need to do so as the matrix corresponding to the original graph had a smaller dimension. The coupling matrix is 
$$
  U = \begin{pmatrix}
  -1 & \phantom{-}0 & \phantom{-}0 & \phantom{-}0 & \phantom{-}0 & \phantom{-}0 & \phantom{-}0 & \phantom{-}0 \\
  \phantom{-}0 & -1/3 & \phantom{-}2/3 & \phantom{-}0 & \phantom{-}2/3 & \phantom{-}0 & \phantom{-}0 & \phantom{-}0\\
  \phantom{-}0 & \phantom{-}2/3 & -1/3 & \phantom{-}0 & \phantom{-}2/3 & \phantom{-}0 & \phantom{-}0 & \phantom{-}0\\
  \phantom{-}0 & \phantom{-}0 & \phantom{-}0 & \phantom{-}0 & \phantom{-}0 & \phantom{-}0 & \phantom{-}1 & \phantom{-}0\\
  \phantom{-}0 & \phantom{-}2/3 & \phantom{-}2/3 & \phantom{-}0 & -1/3 & \phantom{-}0 & \phantom{-}0 & \phantom{-}0\\
  \phantom{-}0 & \phantom{-}0 & \phantom{-}0 & \phantom{-}0 & \phantom{-}0 & \phantom{-}0 & \phantom{-}0 & \phantom{-}1\\
  \phantom{-}0 & \phantom{-}0 & \phantom{-}0 & \phantom{-}1 & \phantom{-}0 & \phantom{-}0 & \phantom{-}0 & \phantom{-}0\\
  \phantom{-}0 & \phantom{-}0 & \phantom{-}0 & \phantom{-}0 & \phantom{-}0 & \phantom{-}1 & \phantom{-}0 & \phantom{-}0
  \end{pmatrix}\,.
$$

Let us assume constant dampings $a_2$ on the cycle and $a_{1}$ on the
attached edge with, in general, $a_1 \ne a_2$. 
%%-----
\begin{figure}[h]
   \begin{center}
   \includegraphics{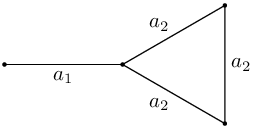}
     \caption{Graph with a cycle and one appended edge (Example \ref{ex-loop1ap})}\label{fig2}
   \end{center}
\end{figure}
\begin{figure}[h]
   \begin{center}
   \includegraphics[width=\textwidth]{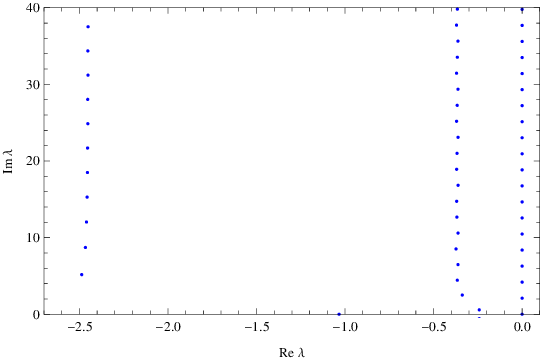}
   \caption{Spectrum of the graph in figure \ref{fig2} for $a_1 = 3$, $a_2 = 0$, $b_1 = 0$, $b_2 = 0$ in Example \ref{ex-loop1ap}}\label{ex72_1}
   \end{center}
\end{figure}
\begin{figure}[h]
   \begin{center}
   \includegraphics[width=\textwidth]{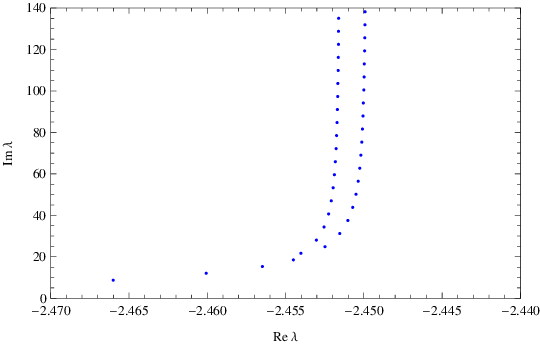}
   \caption{Spectrum of the graph in figure \ref{fig2} for $a_1 = 3$, $a_2 = 0$, $b_1 = 0$, $b_2 = 0$ in Example \ref{ex-loop1ap}, detail}\label{ex72_2}
   \end{center}
\end{figure} 
\begin{figure}[h]
   \begin{center}
   \includegraphics[width=\textwidth]{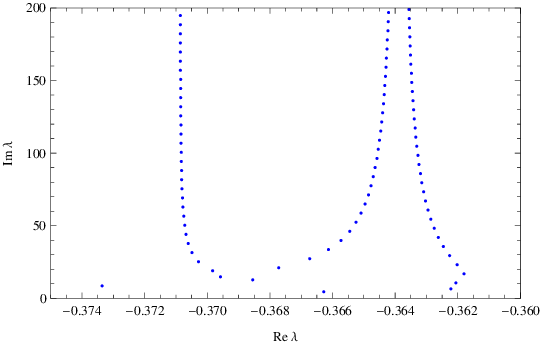}
   \caption{Spectrum of the graph in figure \ref{fig2} for $a_1 = 3$, $a_2 = 0$, $b_1 = 0$, $b_2 = 0$ in Example \ref{ex-loop1ap}, detail}\label{ex72_3}
   \end{center}
\end{figure} 
%%-----

Using the Ansatz $f_1(x) = \alpha_1 \sinh{(\tilde \lambda_1(\lambda) x)}$ on the appendix and $f_2(x) =\alpha_2 \sinh{(\tilde \lambda_2(\lambda)
x)}+\beta_2 \cosh{(\tilde \lambda_2(\lambda) x)}$ on the cycle with $x = 0$ at its center, one has from the continuity of the function on the cycle
%%-----
$$
  f_2(3/2) = f_2(-3/2)\quad \Rightarrow \quad \alpha_2 \sin{\left(\frac{3}{2}\tilde \lambda_2(\lambda)\right)} = 0\,.
$$
%%-----
In a similar way to the previous example, this leads to equation (\ref{eq-ex1-2}) and to the same behaviour 
as in the case of a segment of length $3/2$, \emph{i.e.}
$c_0^{(6)} = -a_2$,
$c_0^{(7)} = -a_2+ \frac{2}{3}\pi i$ and $c_0^{(8)} = -a_2+\frac{4}{3}\pi i$.

Under the assumption $\alpha_2 = 0$ one has
%%-----
$$
  \tilde \lambda_1(\lambda) \cosh{\frac{3\tilde \lambda_2(\lambda)}{2}}\cosh{(\tilde \lambda_1(\lambda))} + 2 \tilde \lambda_2(\lambda)
\sinh{\frac{3\tilde \lambda_2(\lambda)}{2}} \sinh{(\tilde \lambda_1(\lambda))} = 0\,.
$$
%%-----
Using the asymptotic expansions of $\tilde\lambda_j(\lambda)$ and $\lambda$ in $n$, in a way similar to the previous example, we have 
%%-----
$$
  n \left(z^5 - \frac{1}{3}(z^3 \mathrm{e}^{-2 a_1} + z^2 \mathrm{e}^{-3 a_2}) + \mathrm{e}^{- 2 a_1 - 3 a_2} \right) + \mathcal{O}(1) = 0
$$
%%-----
with $z = \mathrm{e}^{c_0^{(s)}}$, which gives the polynomial equation 
$$
  z^5 - \frac{1}{3}(z^3 \mathrm{e}^{-2 a_1} + z^2 \mathrm{e}^{-3 a_2}) + \mathrm{e}^{- 2 a_1 - 3 a_2}  = 0
$$
yielding the high-frequency abscissas. Hence there are in general five sequences of eigenvalues with eigenfunctions having nontrivial
support on the appendix
plus three sequences with eigenfunctions behaving as $\sinh{(\tilde \lambda_2(\lambda)x)}$ on the cycle and having trivial component on the 
segment.
For the particular choice $a_1 = 3$ and $a_2 = 0$ (\emph{i.e.} no damping on the cycle) we have these numerically found roots of the above equations:
$z_{1,2} = - 0.345 \pm 0.603 \,\mathrm{i}$, $z_{3,4} = \pm 0.0863$, $z_{5} = 0.690$. The values of the coefficients are $c_0^{(1,2)} = -0.364 \pm
2.091\,\mathrm{i}$, $c_0^{(3)} = -2.452 + \pi \mathrm{i}$, $c_0^{(4)} = -2.450$, $c_0^{(5)} = -0.371$. The eigenvalues may be computed numerically
in this case and their location is shown in figures~\ref{ex72_1} -- \ref{ex72_3}. The upper bound on the number of high-frequency abscissas given
by Lemma~\ref{thm-seq1} is eight.
}\end{example}
%%-----

%%-----
\begin{example}\label{ex-loop2ap} {\rm \textbf{One cycle with two appended edges} \\
The third example we study consists of five edges, three of them in a cycle, with the remaining two attached at different vertices as shown in
figure \ref{fig3}. Let us again assume standard coupling at vertices connecting two and more edges and Dirichlet coupling at the boundary. The
lengths of all edges are equal to one.
The coupling matrix for this graph is
$$
  U = \begin{pmatrix}
  -1 & \phantom{-}0 & \phantom{-}0 & \phantom{-}0 & \phantom{-}0 & \phantom{-}0 & \phantom{-}0 & 0 & \phantom{-}0 & 0\\
  \phantom{-}0 & -1/3 & \phantom{-}2/3 & \phantom{-}0 & \phantom{-}0 & \phantom{-}0 & \phantom{-}2/3 & 0 & \phantom{-}0 & 0\\
  \phantom{-}0 & \phantom{-}2/3 & -1/3 & \phantom{-}0 & \phantom{-}0 & \phantom{-}0 & \phantom{-}2/3 & 0 & \phantom{-}0 & 0\\
  \phantom{-}0 & \phantom{-}0 & \phantom{-}0 & -1/3 & \phantom{-}2/3 & \phantom{-}0 & \phantom{-}0 & 0 & \phantom{-}2/3 & 0\\
  \phantom{-}0 & \phantom{-}0 & \phantom{-}0 & \phantom{-}2/3 & -1/3 & \phantom{-}0 & \phantom{-}0 & 0 & \phantom{-}2/3 & 0\\
  \phantom{-}0 & \phantom{-}0 & \phantom{-}0 & \phantom{-}0 & \phantom{-}0 & -1 & \phantom{-}0 & 0 & \phantom{-}0 & 0\\
  \phantom{-}0 & \phantom{-}2/3 & \phantom{-}2/3 & \phantom{-}0 & \phantom{-}0 & \phantom{-}0 & -1/3 & 0 & \phantom{-}0 & 0\\
  \phantom{-}0 & \phantom{-}0 & \phantom{-}0 & \phantom{-}0 & \phantom{-}0 & \phantom{-}0 & \phantom{-}0 & 0 & \phantom{-}0 & 1\\
  \phantom{-}0 & \phantom{-}0 & \phantom{-}0 & \phantom{-}2/3 & \phantom{-}2/3 & \phantom{-}0 & \phantom{-}0 & 0 & -1/3 & 0\\
  \phantom{-}0 & \phantom{-}0 & \phantom{-}0 & \phantom{-}0 & \phantom{-}0 & \phantom{-}0 & \phantom{-}0 & 1 & \phantom{-}0 & 0 
  \end{pmatrix}
$$

We use the Ansatz 
$f_{1,2}(x) = \alpha_{1,2}\sinh{\tilde\lambda_{1,2}x}$ on the appendices, $f_{3}(x) = \alpha_3 \sinh{\tilde\lambda_3 x} +  \beta_3
\cosh{\tilde\lambda_3 x}$ on two edges of the cycle with $x = 0$ in the only vertex of degree two and $f_{4}(x) = \alpha_4 \sinh{\tilde\lambda_4 x} + 
\beta_4 \cosh{\tilde\lambda_4 x}$ on the last edge of the cycle with $x = 0$ in its center. For simplicity we omit explicit dependence of $\tilde
\lambda_{j}$ on $\lambda$. From the coupling conditions we get
%%-----
\begin{eqnarray*}
  \alpha_1 \sinh{\tilde \lambda_1 x} = \alpha_4 \sinh{\frac{\tilde \lambda_3}{2}} + \beta_4 \cosh{\frac{\tilde \lambda_3}{2}} = \alpha_3 \sinh{\tilde
\lambda_3} + \beta_3 \cosh{\tilde \lambda_3}\,,\\
  \alpha_2 \sinh{\tilde \lambda_2 x} = -\alpha_4 \sinh{\frac{\tilde \lambda_3}{2}} + \beta_4 \cosh{\frac{\tilde \lambda_3}{2}} = -\alpha_3
\sinh{\tilde \lambda_3} + \beta_3 \cosh{\tilde \lambda_3}\,,\\  
  \tilde \lambda_1 \alpha_1 \cosh{\tilde \lambda_1} + \tilde \lambda_3 \left[\alpha_4 \cosh{\frac{\tilde \lambda_3}{2}} + \beta_4 \sinh{\frac{\tilde
\lambda_3}{2}}+ \alpha_3 \cosh{\tilde \lambda_3} + \beta_3 \sinh{\tilde \lambda_3}\right] = 0\,,\\
  \tilde \lambda_2 \alpha_2 \cosh{\tilde \lambda_2} + \tilde \lambda_3 \left[-\alpha_4 \cosh{\frac{\tilde \lambda_3}{2}} + \beta_4 \sinh{\frac{\tilde
\lambda_3}{2}} - \alpha_3 \cosh{\tilde \lambda_3} + \beta_3 \sinh{\tilde \lambda_3}\right] = 0\,.\\
\end{eqnarray*}
%%-----
The determinant of the above system gives the secular equation in this case, and the location of the eigenvalues found numerically for
particular values of the damping coefficients are shown in figures \ref{ex73_1} -- \ref{ex73_2}.
\begin{figure}[h]
   \begin{center}
   \includegraphics{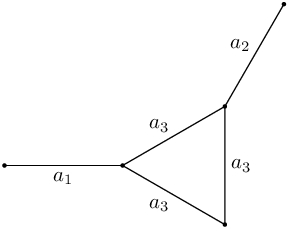}
     \caption{Graph with one cycle and two appended edges (Example \ref{ex-loop2ap})}\label{fig3}
   \end{center}
\end{figure} 

\begin{figure}[h]
   \begin{center}
   \includegraphics[width=\textwidth]{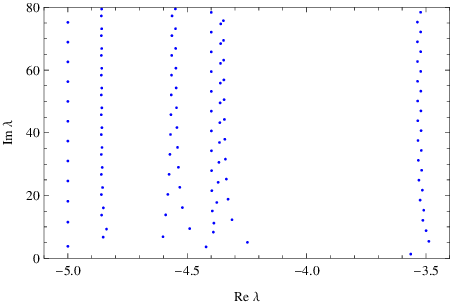}
     \caption{Spectrum of the graph in figure \ref{fig3} for $a_1 = 3$, $a_2 = 4$, $a_3 = 5$, $b_1 = b_2 = b_3 = 0$ in Example \ref{ex-loop2ap}}\label{ex73_1}
   \end{center}
\end{figure} 

\begin{figure}[h]
   \begin{center}
   \includegraphics[width=\textwidth]{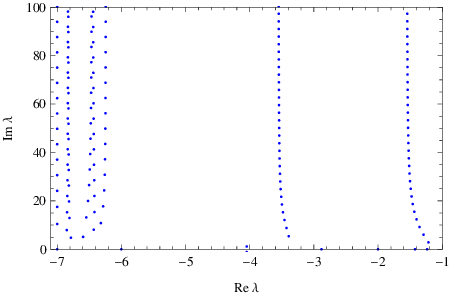}
     \caption{Spectrum of the graph in figure \ref{fig3} for $a_1 = 1$, $a_2 = 3$, $a_3 = 7$, $b_1 = b_2 = b_3 = 0$ in Example \ref{ex-loop2ap}}\label{ex73_2}
   \end{center}
\end{figure} 

Using the asymptotic expansion of $\tilde \lambda_j(\lambda)$ in $n$ one obtains a polynomial equation in $z = \mathrm{e}^{c_0^{(s)}}$.
\begin{multline*}
  -9 + (3\,\mathrm{e}^{2 a_1} + 3\,\mathrm{e}^{2 a_2}+ \mathrm{e}^{2 a_3})z^2+8\mathrm{e}^{3a_3}z^3+(\mathrm{e}^{4a_3}+\mathrm{e}^{2(a_1+a_3)}-\mathrm{e}^{2(a_1+a_2)}+
  \\
  +\mathrm{e}^{2(a_2+a_3)}) z^4 - 8(\mathrm{e}^{2a_1+3a_3}+\mathrm{e}^{2a_2+3a_3})z^5+(\mathrm{e}^{2(a_1+a_2+a_3)}-\mathrm{e}^{6a_3}+\mathrm{e}^{2a_1+4a_3}+
  \\
  +\mathrm{e}^{2a_2+4a_3})z^6+8 \mathrm{e}^{2a_1+2a_2+3a_3}z^7+(\mathrm{e}^{2a_1+2a_2+4a_3}+3\mathrm{e}^{2a_1+6a_3}+3\mathrm{e}^{2a_1+6a_3})z^8-
  \\
  -9\mathrm{e}^{2a_1+2a_2+6a_3}z^{10} = 0\,.
\end{multline*}
For the choice of the damping $a_1 = 3$, $a_2 = 4$ and $a_3 = 5$ it has roots
\begin{eqnarray*}
 z_1 = -0.02951, \quad z_2 = -0.01228, \quad z_3 =  -0.00471- 0.00938 \,\mathrm{i},\\
 z_4 =  -0.00471+ 0.00938\,\mathrm{i}, \quad z_5 = -0.00354 - 0.00690 \,\mathrm{i},\\ 
 z_6 = -0.00354 + 0.00690 \,\mathrm{i},\quad z_7 =  0.00674, \quad z_8 = 0.01122 - 0.00626\,\mathrm{i},\\
 z_9 = 0.01122 - 0.00626\,\mathrm{i},\quad  z_{10} = 0.02911.
\end{eqnarray*}
This leads to eigenvalues with
\begin{eqnarray*}
c_0^{(1)} = -3.52 + 3.14\,\mathrm{i}, &\quad c_0^{(2)} = -4.40 + 3.14 \,\mathrm{i}, &\quad c_0^{(3)} = -4.56 - 2.04\,\mathrm{i},\\
c_0^{(4)} = -4.56 + 2.04\,\mathrm{i}, &\quad c_0^{(5)} =  -4.86 - 2.05 \,\mathrm{i}, &\quad c_0^{(6)} = -4.86 + 2.05  \,\mathrm{i},\\
c_0^{(7)} = -5, &\quad c_0^{(8)} = -4.35 - 0.51\,\mathrm{i}, &\quad c_0^{(9)} = -4.35 + 0.51\,\mathrm{i}, \\
c_0^{(10)} = -3.54.
\end{eqnarray*}
The upper bound on the number of high-frequency abscissas according to Lemma~\ref{thm-seq1} is ten.
}\end{example}

We showed in Section~\ref{sec-number} that there can be at most $2N$ distinct high-frequency abscissas. We shall now present an example of a
graph with standard coupling where this maximum is attained.

\begin{example}\label{ex-complete}{\rm \textbf{Complete graph on four vertices $K_{4}$}\\
Let us consider a complete graph $K_{4}$ with all edge lengths equal
to one (see Figure~\ref{fig4}). The secular equation can be obtained by the approach shown in the previous sections
or in Sections~\ref{sec-seceq} and~\ref{sec-examples}.

\begin{figure}[h]
   \begin{center}
   \includegraphics[width=5cm]{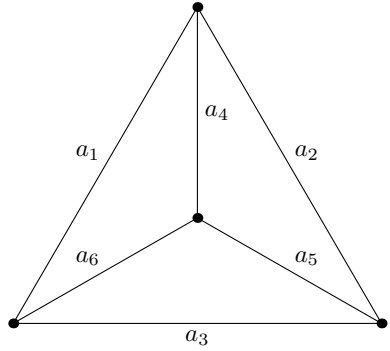}
     \caption{Complete graph on four vertices (Example~\ref{ex-complete})}\label{fig4}
   \end{center}
\end{figure} 

For a special choice of damping coefficients $a_1 = 7.7$, $a_2 = 10.5$, $a_3 = 13.7$, $a_4 = 13.7$, $a_5 = 5.7$, $a_6 = 11.2$, one finds that the
secular equation leads to the following polynomial in $z = \mathrm{e}^{c_0^{(s)}}$
\begin{multline*}
  -81 + 1.43471\cdot 10^{13} z^2 + 1.17346\cdot 10^{16} z^3 - 5.30758\cdot 10^{23} z^4 -3.11505 \cdot 10^{27} z^5 + 
  \\
  + 4.21374 \cdot 10^{33} z^6 - 8.24732 \cdot 10^{36} z^7 - 4.34523 \cdot 10^{42} z^8 + 1.63257 \cdot 10^{43} z^9 +
  \\ + 1.98616 \cdot 10^{50} z^{10} -  1.56782 \cdot 10^{56} z^{12} = 0\,,
\end{multline*}
which has (approximate) roots 
\begin{eqnarray*}
 z_1 = -0.00112, \quad z_2 = -0.000145, \quad z_3 =  -0.0000312,\quad z_4 = -9.53\cdot 10^{-6},\\
 \quad z_5 = -5.42\cdot 10^{-6}, \quad z_6 = -2.77 \cdot 10^{-6}, \quad z_7 =  2.78 \cdot 10^{-6}, \quad z_8 = 5.02\cdot 10^{-6},\\
  z_9 = 0.0000113, \quad z_{10} = 0.0000276, \quad z_{11} = 0.000147, \quad z_{12} = 0.00112.
\end{eqnarray*}
This leads to eigenvalues with 
\begin{eqnarray*}
c_0^{(1)} = -6.80 + 3.14 \,\mathrm{i}, &\quad c_0^{(2)} = -8.84 + 3.14 \,\mathrm{i}, &\quad c_0^{(3)} = -10.4 + 3.14 \,\mathrm{i},\\
c_0^{(4)} = -11.6 + 3.14 \,\mathrm{i}, &\quad c_0^{(5)} = -12.1 + 3.14 \,\mathrm{i}, &\quad c_0^{(6)} = -12.797 + 3.14 \,\mathrm{i},\\
c_0^{(7)} = -12.792, &\quad c_0^{(8)} = -12.2, &\quad c_0^{(9)} = -11.4, \\
c_0^{(10)} = -10.5, &\quad c_0^{(11)} = -8.83, &\quad c_0^{(12)} = -6.80.
\end{eqnarray*}
We can see that the maximal number of distinct high-frequency abscissas, which is 12, is attained.
}\end{example}

\begin{example}\label{ex-star}{\rm \textbf{Star graph with different edge lengths}\\
Now we illustrate what happens when the lengths of the edges of the graph are changed. We consider a star graph with three edges, Dirichlet coupling at
the free ends and standard coupling in the central vertex. The coupling matrix is
$$
  U = \begin{pmatrix}
  -1 & \phantom{-}0 & \phantom{-}0 & \phantom{-}0 & \phantom{-}0 & \phantom{-}0\\
  \phantom{-}0 & -1/3 & \phantom{-}0 & \phantom{-}2/3 & \phantom{-}0 & \phantom{-}2/3\\
  \phantom{-}0 & \phantom{-}0 & -1 & \phantom{-}0 & \phantom{-}0 & \phantom{-}0\\ 
  \phantom{-}0 & \phantom{-}2/3 & \phantom{-}0 & -1/3 & \phantom{-}0 & \phantom{-}2/3\\
  \phantom{-}0 & \phantom{-}0 & \phantom{-}0 & \phantom{-}0 & -1 & \phantom{-}0\\
  \phantom{-}0 & \phantom{-}2/3 & \phantom{-}0 & \phantom{-}2/3 & \phantom{-}0 & -1/3
  \end{pmatrix}
$$
If all the lengths are commensurate then the graph can be described by the machinery of
previous sections, where the unit length is the greatest common divisor of the edge lengths.

Using the Ansatz $f_j(x) = \alpha_j \sinh{\tilde \lambda_j x}$ on each edge with $x = 0$ at the free end one obtains the secular equation
$$
   \sum_{j=1}^{3} \tilde\lambda_j \cosh{\left(\tilde{\lambda_j} l_j\right)} \prod_{\stackrel{i = 1}{i \ne j}}^{3} \sinh{\left(\tilde{\lambda_i} l_i\right)} = 0\,. 
$$
The eigenvalue location for particular lengths of the edges and the choice $a_1 = 3$,  $a_2 = 4$,  $a_3 = 5$ are shown in figures
\ref{exl_1}--\ref{exl_3}. From this we see that by choosing different lengths for the edges it will be possible to increase the number of high-frequency
abscissas and make it larger than the value of $2N$ given by Lemmata~\ref{thm-seq1} and~\ref{thm-nseq} for the equilateral case.
\begin{figure}[h]
   \begin{center}
   \includegraphics[width=\textwidth]{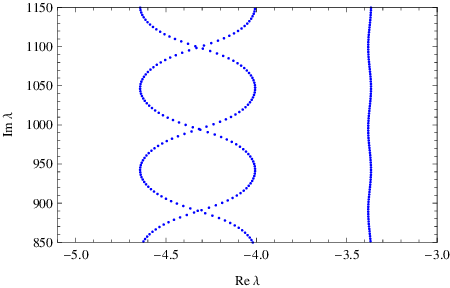}
     \caption{Spectrum of a star graph with different lengths of the edges, $l_1 = 1$, $l_2 = 1$, $l_3 = 1.03$ in Example~\ref{ex-star}}
     \label{exl_1}
   \end{center}
\end{figure}
\begin{figure}[h]
   \begin{center}
   \includegraphics[width=\textwidth]{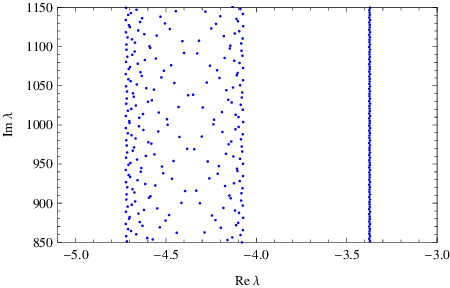}
     \caption{Spectrum of a star graph with different lengths of the edges, $l_1 = 1$, $l_2 = 1$, $l_3 = 1.41$ in Example~\ref{ex-star}}
     \label{exl_2}
   \end{center}
\end{figure}
\begin{figure}[h]
   \begin{center}
   \includegraphics[width=\textwidth]{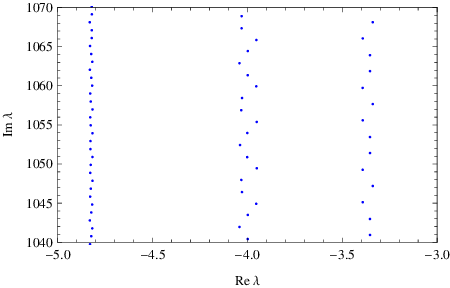}
     \caption{Spectrum of a star graph with different lengths of the edges, $l_1 = 1.5$, $l_2 = 2.1$, $l_3 = 3.1$ in Example~\ref{ex-star}}
     \label{exl_3}
   \end{center}
\end{figure}

Finally, we study the behaviour of the spectra of this graph for $l_1 = l_2 = 1$, $l_3 = 1+\varepsilon$, $a_1 (x) = a_2(x) = 1$,
$a_3 (x) = 1+ a\varepsilon$
for small $\varepsilon$ and fixed $a$. For $\varepsilon = 0$ we have eigenvalues given by the equation $\sinh{\tilde \lambda_1} = 0$ (with multiplicity 2)
and eigenvalues given by $\cosh{\tilde \lambda_1} = 0$ (with multiplicity 1). For nonzero $\varepsilon$ there is a sequence of eigenvalues given by
$\sinh{\tilde \lambda_1} = 0$ (with multiplicity 1) and we find from the numerics that, starting with real parts of the eigenvalues approximately $-1-a/3$
and $-1-2a/3$, begin to interlace each other for bigger imaginary part of the eigenvalues (see Figure~\ref{ex75_4}). The imaginary part of the point, 
where these two sequences first cross, grows as $\varepsilon$ approaches zero. The difference
between the imaginary parts of two consecutive eigenvalues in 
this sequence is approximately $\pi$. Hence, we conjecture that for $\varepsilon$ rational and approaching zero the number of high-frequency abscissas
grows to infinity and for $\varepsilon$ irrational the measure $\mu_\infty$ defined in Section~\ref{sec-number} is no longer atomic.

\begin{figure}[h]
   \begin{center}
   \includegraphics[width=\textwidth]{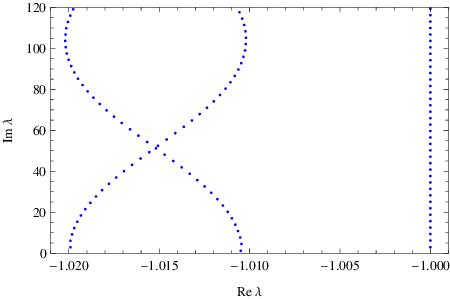}
     \caption{Spectrum of a star graph with different lengths of the edges, $l_1 = 1$, $l_2 = 1$, $l_3 = 1.03$, $a_1 = a_2 = 1$, $a_3 = 1.03$ in Example~\ref{ex-star}}
     \label{ex75_4}
   \end{center}
\end{figure}

}\end{example}

%\bibliographystyle{qg}
%\bibliography{qg}

\newcommand{\etalchar}[1]{$^{#1}$}

\end{document}